\documentclass[11pt]{article}

\usepackage[utf8]{inputenc}
\usepackage{lmodern}
\usepackage{times}
\usepackage[margin=1in]{geometry}

\usepackage{amsmath, amsthm, amssymb}
\usepackage{mathtools}
\usepackage{cite}
\usepackage{url}
\usepackage{hyperref}
\usepackage{cleveref}

\usepackage{appendix}
\usepackage{graphicx}
\usepackage{epstopdf}
\usepackage{enumerate}
\usepackage{color}
\usepackage{xcolor}
\usepackage{algorithm}
\usepackage[noend]{algpseudocode}
\usepackage{multirow}
\usepackage{paralist}
\usepackage{xspace}
\usepackage[shortlabels]{enumitem}
\usepackage{nicefrac}
\usepackage{mathtools}
\usepackage{xparse}
\usepackage{etoolbox}
\usepackage{thmtools, thm-restate}

\usepackage{tabto}

\usepackage{ifthen}
\usepackage{arrayjobx}
\usepackage{tikz}
\usetikzlibrary{positioning,decorations.pathreplacing}

\newtoggle{shortversion}
\toggletrue{shortversion}

\def\includeHeader{}%
\def\includeTail{}%
\hypersetup{
	colorlinks=true,                  
	linkcolor=black
}

\iftoggle{shortversion}
{
  \let\myenumerate\compactenum
  \let\endmyenumerate\endcompactenum
  \let\myitemize\compactitem
  \let\endmyitemize\endcompactitem
}{
  
  \let\myenumerate\enumerate
  \let\endmyenumerate\endenumerate
  \let\myitemize\itemize
  \let\endmyitemize\enditemize
}

\newcommand{\loweredqedbox}{\raisebox{-0.5\baselineskip}{\llap{\openbox}}}

\newcommand{\theW}{\mathcal{W}}

\newcommand{\sebastian}[1]{[[[ \textcolor{green}{\bf Seb:} {\em #1} ]]]}
\newcommand{\yannic}[1]{[[[ \textcolor{red}{\bf Yannic:} {\em #1} ]]]}
\newcommand{\fabian}[1]{[[[ \textcolor{orange}{\bf Fabian:  #1} ]]]}
\newcommand{\christof}[1]{[[[ \textcolor{orange}{\bf Christof:} {\em #1} ]]]}

\newcommand{\mtwo}{l}
\newcommand{\cnenner}{5c}
\newcommand{\NB}{\textit{NB}}

\newcommand{\NN}{\ensuremath{\mathbb{N}}}
\newcommand{\RR}{\ensuremath{\mathbb{R}}}
\newcommand{\cl}{\mathbf{cl}}
\newcommand{\SSS}{\mathcal{S}}
\newcommand{\CCC}{\mathcal{C}}
\newcommand{\TTT}{\mathcal{T}}
\newcommand{\aEvent}{\mathcal{A}}
\newcommand{\cEvent}{\mathcal{C}}
\newcommand{\eEvent}{\mathcal{E}}
\newcommand{\xEvent}{\mathcal{X}}
\newcommand{\dbtfull}{leaf-connected tree\xspace}
\newcommand{\dbtshort}{\text{LCT}\xspace}
\newcommand{\dbts}{\text{LCTs}\xspace}
\DeclareDocumentCommand \dbt { g } {\IfNoValueTF{#1}{\ensuremath{\dbtshort}\xspace}{\ensuremath{#1\text{-}\dbtshort}\xspace}}
\newcommand{\dcup}{\ensuremath{\mathaccent\cdot\cup}}

\newcommand{\alg}{\ensuremath{\textsc{ALG}}\xspace}

\newcommand{\pull}{\ensuremath{\textsc{PULL}}\xspace}
\newcommand{\push}{\ensuremath{\textsc{PUSH}}\xspace}
\newcommand{\rpull}{\ensuremath{\textsc{RPULL}}\xspace}
\newcommand{\vpull}{\ensuremath{\textsc{VPULL}}\xspace}
\newcommand{\pushpull}{\push-\pull}
\newcommand{\pushrpull}{\push-\rpull}

\newcommand{\roundpull}{\ensuremath{\textsc{PULL}_1}}
\newcommand{\roundrpull}{\ensuremath{\textsc{RPULL}_{T}}}
\newcommand{\roundvpull}{\ensuremath{\textsc{VPULL}_{T+1}}}

\newcommand{\boldroundpull}{\ensuremath{\mathbf{PULL}_1}}
\newcommand{\boldroundrpull}{\ensuremath{\mathbf{RPULL}_{T}}}
\newcommand{\boldroundvpull}{\ensuremath{\mathbf{VPULL}_{T+1}}}
\newcommand{\boldroundvpullT}{\ensuremath{\mathbf{VPULL}_{T}}}

\newcommand{\boldpull}{\ensuremath{\mathbf{PULL}}}
\newcommand{\boldpush}{\ensuremath{\mathbf{PUSH}}}
\newcommand{\boldrpull}{\ensuremath{\mathbf{RPULL}}}
\newcommand{\boldvpull}{\ensuremath{\mathbf{VPULL}}}

\newcommand{\Srpull}{\ensuremath{S_{T}^{\rpull}}}
\newcommand{\Svpull}{\ensuremath{S_{T+1}^{\vpull}}}
\newcommand{\Spull}{\ensuremath{S_{1}^{\pull}}}
\newcommand{\Srpullz}{\ensuremath{X}}
\newcommand{\Svpullz}{\ensuremath{X^{\textsc{VP}}}}
\newcommand{\Spullz}{\ensuremath{X^\textsc{P}}}

\newcommand{\prp}{\pushpull}
\newcommand{\GoodExecution}{\ensuremath{\mathbb{G}}}

\newcommand{\ddelta}{\ensuremath{\frac{\Delta}{\delta}}}
\newcommand{\ddlogn}{\ensuremath{\frac{\Delta}{\delta}\log n}}
\newcommand\lcb{\left\lbrace} 
\newcommand\rcb{\right\rbrace} 
\newcommand\cb[1]{\lcb #1 \rcb} 
\newcommand\lab{\left[} 
\newcommand\rab{\right]} 
\newcommand\lb{\left(} 
\newcommand\rb{\right)} 
\newcommand\br[1]{\lb #1 \rb} 

\newcommand{\rv}[1]{\ensuremath{\mathrm{#1}}} 
\newcommand{\srv}[1]{\ensuremath{\mathcal{#1}}} 

\newcommand{\card}[1]{\left\vert{#1}\right\vert}

\DeclareMathOperator{\E}{\mathbf{E}}
\let\Pr\relax
\DeclareMathOperator{\Pr}{\mathbf{P}}
\DeclareMathOperator{\powset}{\mathfrak{P}}

\DeclareMathOperator{\polylog}{\ensuremath{\mathrm{polylog}}}
\DeclareMathOperator{\polyloglog}{\ensuremath{\mathrm{polyloglog}}}

\newcommand{\EOf}[1]{\E\mathopen{}\lab #1 \rab\mathclose{}}
\newcommand{\PrOf}[1]{\Pr\mathopen{}\lb #1 \rb\mathclose{}}
\newcommand{\powsetOf}[1]{\powset\mathopen{}\lb #1 \rb\mathclose{}}

\newcommand{\bigO}{\ensuremath{O}}
\newcommand{\bigOOf}[1]{\bigO \mathopen{}\lb #1 \rb\mathclose{}}
\newcommand{\smallO}{\ensuremath{o}}
\newcommand{\smallOOf}[1]{\smallO \mathopen{}\lb #1 \rb\mathclose{}}
\newcommand{\OmegaOf}[1]{\Omega \mathopen{}\lb #1 \rb\mathclose{}}
\newcommand{\omegaOf}[1]{\omega \mathopen{}\lb #1 \rb\mathclose{}}
\newcommand{\ThetaOf}[1]{\Theta \mathopen{}\lb #1 \rb\mathclose{}}

\newcommand{\ID}{\textit{ID}}

\newcommand{\local}{\mathcal{LOCAL}}
\newcommand{\col}{\top}
\newtheorem{theorem}{Theorem}[section]
\newtheorem{lemma}[theorem]{Lemma}
\newtheorem{claim}[theorem]{Claim}
\newtheorem{corollary}[theorem]{Corollary}
\newtheorem{definition}[theorem]{Definition}

\newcommand{\false}{\ensuremath{\mathbf{false}}}
\newcommand{\true}{\ensuremath{\mathbf{true}}}
\newcommand{\Sstate}{\texttt{informed}\xspace}
\newcommand{\Ustate}{\texttt{uninformed}\xspace}
\newcommand{\state}{\ensuremath{\textit{state}_v}}
\newcommand{\waiting}{\ensuremath{\textit{tokenReceived}_v}}
\newcommand{\ack}{\textit{token}}
\newcommand{\badexec}{\textit{BE}}
\newcommand{\strongly}{\textit{sc}}
\newcommand{\msg}{\textit{msg}}
\newcommand{\textpath}{\textit{path}}

\algnewcommand\algorithmicswitch{\textbf{switch}}
\algnewcommand\algorithmiccase{\textbf{case}}

\algdef{SE}[SWITCH]{Switch}{EndSwitch}[1]{\algorithmicswitch\ #1\ \algorithmicdo}{\algorithmicend\ \algorithmicswitch}%
\algdef{SE}[CASE]{Case}{EndCase}[1]{\algorithmiccase\ #1}{\algorithmicend\ \algorithmiccase}%
\algtext*{EndSwitch} 
\algtext*{EndCase} 
\algdef{SE}[WithProb]{WithProb}{EndWithProb}[1]{\algorithmicwithprob\ #1\ \algorithmicdo}{\algorithmicend\ \algorithmicwithprob}%
\algdef{Ce}[Otherwise]{WithProb}{Otherwise}{EndWithProb}{\algorithmicotherwise}%
\algtext*{EndWithProb}%
\algtext*{EndOtherwise}%

\algnewcommand{\NoAlignComment}[1]{\hspace{1cm}\textit{// #1}}
\algnewcommand{\FullLineComment}[1]{\Statex\textit{// #1}}
\algnewcommand\algorithmicwithprob{\textbf{with probability}}
\algnewcommand\algorithmicotherwise{\textbf{otherwise}}
\algnewcommand\algorithmicdbupdate{\textbf{update}}
\algnewcommand{\DBUpdate}[1]{\State \algorithmicdbupdate\ #1\ \textbf{with: }}
\algnewcommand{\algorithmicgoto}{\textbf{go to}}%
\algnewcommand{\Goto}[1]{\State \algorithmicgoto~\ref{#1}}%
\newcommand{\otherwise}{\textbf{, otherwise }}

\newcommand{\eps}{\epsilon}
\newcommand{\logstar}[1]{\log^{\ast}{#1}}
\newcommand{\poly}[1]{\ensuremath{\mathrm{poly}({#1})}}
\newcommand{\Delt}{\Delta}

\newcommand{\TODO}[1]{ {\color{blue} TODO: #1 }}
\newcommand{\tocheck}[1]{{\color{red} #1}}
\newcommand{\strikeout}[1]{\hide{#1}}
\newcommand{\gone}[1]{{\color{blue}\sout{#1}}}
\newcommand{\added}[1]{{\color{red}#1}}
\newcommand{\hide}[1]{}

\renewcommand{\vec}[1]{\underline{#1}}
\newcommand{\dotcup}{\dcup}
\newcommand{\code}[1]{\texttt{#1}}
\newcommand{\norm}[1]{\lVert #1 \rVert}
\newcommand{\set}[1]{\ensuremath{\left\{#1\right\}}}
\newcommand{\bigbrackets}[1]{\ensuremath{\big(#1\big)}}
\newcommand{\brackets}[1]{\ensuremath{\left(#1\right)}}
\newcommand{\st}{\medspace | \medspace}
\newcommand{\coloneq}{\mathrel{\mathop:}=}

\newcommand{\whpof}[1]{\text{w.h.p.}(\ensuremath{#1})}
\newcommand{\Whpof}[1]{\text{W.h.p.}(\ensuremath{#1})}
\DeclareDocumentCommand \whp { g } {\IfNoValueTF{#1}{\ensuremath{\text{w.h.p.}}\xspace}{\ensuremath{\text{w.h.p.}(#1)}\xspace}}
\DeclareDocumentCommand \Whp { g } {\IfNoValueTF{#1}{\ensuremath{\text{W.h.p.}}\xspace}{\ensuremath{\text{W.h.p.}(#1)}\xspace}}
\newcommand{\wcp}{\text{w.c.p.}\xspace}
\newcommand{\Wcp}{\text{W.c.p.}\xspace}
\newcommand{\wvhp}[1]{\text{w.v.h.p.}(\ensuremath{#1})\xspace}
\newcommand{\Wvhp}[1]{\text{W.v.h.p.}(\ensuremath{#1})\xspace}
\newcommand{\wrt}{\text{with regard to}\xspace}

\newcommand{\Section}[1]{\vspace{-.0in}\section{#1}\vspace{-.0in}}
\newcommand{\Paragraph}[1]{\vspace{.0in}\noindent {{\bf #1}\quad}}
\newcommand{\SectionS}[1]{\vspace{0in}\section*{#1}\vspace{0in}}
\newcommand{\Subsection}[1]{\vspace{-.0in}\subsection{#1}\vspace{-.0in}}
\newcommand{\Subsubsection}[1]{\vspace{0in}\subsubsection{#1}\vspace{0in}}
\newcommand{\paraclose}[1]{\vspace{0.2em}\noindent\textbf{#1.}~}

\newcommand{\FullOrShort}{short}

\ifthenelse{\equal{\FullOrShort}{full}}{
	  
	  \newcommand{\fullOnly}[1]{#1}
	  \newcommand{\shortOnly}[1]{}

  }{
	 \newcommand{\fullOnly}[1]{}
	 \newcommand{\shortOnly}[1]{#1}
  }

\newcommand{\todo}[1]{{\color{red} TODO: #1}}
\newcommand{\remark}[2][]
{{\color{green}\ifthenelse{\equal{#1}{}}{Remark:}{Remark (#1):} #2}}

\date{}
\title{Rumor Spreading with Bounded In-Degree}

\newcommand{\authornamestyle}[1]{#1}
\newcommand{\authordeptstyle}[1]{{\small #1}}
\newcommand{\authorunivstyle}[1]{{\small #1}}
\newcommand{\authorcitystyle}[1]{{\small #1}}
\newcommand{\authormailstyle}[1]{{\small #1}}

\iftrue
\author{
  \authornamestyle{
    Sebastian Daum
  }\\ 
  \authordeptstyle{
    Dept.\ of Comp.\ Science
  }\\
  \authorunivstyle{
  U.\ of Freiburg, Germany
  }\\
  \authorcitystyle{
  Freiburg, Germany
  }\\
  \authormailstyle{
    sdaum@cs.uni-freiburg.de
  }
  \and %
  \authornamestyle{
    Fabian Kuhn
  }\\ 
  \authordeptstyle{
    Dept.\ of Comp.\ Science
  }\\
  \authorunivstyle{
    U.\ of Freiburg, Germany
  }\\
  \authorcitystyle{
    Freiburg, Germany
  }\\
  \authormailstyle{
    kuhn@cs.uni-freiburg.de
  }
  \and %
  \authornamestyle{
    Yannic Maus
  }\\ 
  \authordeptstyle{
    Dept.\ of Comp.\ Science
  }\\
  \authorunivstyle{
    U.\ of Freiburg, Germany
  }\\
  \authorcitystyle{
    Freiburg, Germany
  }\\
  \authormailstyle{
    yannic.maus@cs.uni-freiburg.de
  }
}
\fi

\hide{
\author{ 
\alignauthor
	Sebastian Daum\\
        \affaddr{University of Freiburg}\\
	\affaddr{Freiburg, Germany}\\
	\email{\mbox{sdaum@cs.uni-freiburg.de}}\\
\alignauthor
	Fabian Kuhn\\
        \affaddr{University of Freiburg}\\
	\affaddr{Freiburg, Germany}\\
	\email{kuhn@cs.uni-freiburg.de}\\
\alignauthor
	Christof Sch\"otz\\
	\affaddr{University of Freiburg}\\
	\affaddr{Freiburg, Germany}\\
	\email{schoetzc@informatik.uni-freiburg.de}
}
}

\begin{document}
\setcounter{page}{0}

%
\ifcsname includeHeader\endcsname%
    \def\includeHeader{}%
	\def\includeTail{}%
\else%
    \def\includeHeader{


\date{}
\title{Rumor Spreading with Bounded In-Degree}

\newcommand{\authornamestyle}[1]{#1}
\newcommand{\authordeptstyle}[1]{{\small #1}}
\newcommand{\authorunivstyle}[1]{{\small #1}}
\newcommand{\authorcitystyle}[1]{{\small #1}}
\newcommand{\authormailstyle}[1]{{\small #1}}

\iftrue
\author{
  \authornamestyle{
    Sebastian Daum
  }\\ 
  \authordeptstyle{
    Dept.\ of Comp.\ Science
  }\\
  \authorunivstyle{
  U.\ of Freiburg, Germany
  }\\
  \authorcitystyle{
  Freiburg, Germany
  }\\
  \authormailstyle{
    sdaum@cs.uni-freiburg.de
  }
  \and %
  \authornamestyle{
    Fabian Kuhn
  }\\ 
  \authordeptstyle{
    Dept.\ of Comp.\ Science
  }\\
  \authorunivstyle{
    U.\ of Freiburg, Germany
  }\\
  \authorcitystyle{
    Freiburg, Germany
  }\\
  \authormailstyle{
    kuhn@cs.uni-freiburg.de
  }
  \and %
  \authornamestyle{
    Yannic Maus
  }\\ 
  \authordeptstyle{
    Dept.\ of Comp.\ Science
  }\\
  \authorunivstyle{
    U.\ of Freiburg, Germany
  }\\
  \authorcitystyle{
    Freiburg, Germany
  }\\
  \authormailstyle{
    yannic.maus@cs.uni-freiburg.de
  }
}
\fi

\hide{
\author{ 
\alignauthor
	Sebastian Daum\\
        \affaddr{University of Freiburg}\\
	\affaddr{Freiburg, Germany}\\
	\email{\mbox{sdaum@cs.uni-freiburg.de}}\\
\alignauthor
	Fabian Kuhn\\
        \affaddr{University of Freiburg}\\
	\affaddr{Freiburg, Germany}\\
	\email{kuhn@cs.uni-freiburg.de}\\
\alignauthor
	Christof Sch\"otz\\
	\affaddr{University of Freiburg}\\
	\affaddr{Freiburg, Germany}\\
	\email{schoetzc@informatik.uni-freiburg.de}
}
}


}%
	\def\includeTail{
}%
\fi%
%
\hypersetup{
	colorlinks=true,                  
	linkcolor=black
}

\iftoggle{shortversion}
{
  \let\myenumerate\compactenum
  \let\endmyenumerate\endcompactenum
  \let\myitemize\compactitem
  \let\endmyitemize\endcompactitem
}{
  
  \let\myenumerate\enumerate
  \let\endmyenumerate\endenumerate
  \let\myitemize\itemize
  \let\endmyitemize\enditemize
}

\newcommand{\loweredqedbox}{\raisebox{-0.5\baselineskip}{\llap{\openbox}}}

\newcommand{\theW}{\mathcal{W}}

\newcommand{\sebastian}[1]{[[[ \textcolor{green}{\bf Seb:} {\em #1} ]]]}
\newcommand{\yannic}[1]{[[[ \textcolor{red}{\bf Yannic:} {\em #1} ]]]}
\newcommand{\fabian}[1]{[[[ \textcolor{orange}{\bf Fabian:  #1} ]]]}
\newcommand{\christof}[1]{[[[ \textcolor{orange}{\bf Christof:} {\em #1} ]]]}

\newcommand{\mtwo}{l}
\newcommand{\cnenner}{5c}
\newcommand{\NB}{\textit{NB}}

\newcommand{\NN}{\ensuremath{\mathbb{N}}}
\newcommand{\RR}{\ensuremath{\mathbb{R}}}
\newcommand{\cl}{\mathbf{cl}}
\newcommand{\SSS}{\mathcal{S}}
\newcommand{\CCC}{\mathcal{C}}
\newcommand{\TTT}{\mathcal{T}}
\newcommand{\aEvent}{\mathcal{A}}
\newcommand{\cEvent}{\mathcal{C}}
\newcommand{\eEvent}{\mathcal{E}}
\newcommand{\xEvent}{\mathcal{X}}
\newcommand{\dbtfull}{leaf-connected tree\xspace}
\newcommand{\dbtshort}{\text{LCT}\xspace}
\newcommand{\dbts}{\text{LCTs}\xspace}
\DeclareDocumentCommand \dbt { g } {\IfNoValueTF{#1}{\ensuremath{\dbtshort}\xspace}{\ensuremath{#1\text{-}\dbtshort}\xspace}}
\newcommand{\dcup}{\ensuremath{\mathaccent\cdot\cup}}

\newcommand{\alg}{\ensuremath{\textsc{ALG}}\xspace}

\newcommand{\pull}{\ensuremath{\textsc{PULL}}\xspace}
\newcommand{\push}{\ensuremath{\textsc{PUSH}}\xspace}
\newcommand{\rpull}{\ensuremath{\textsc{RPULL}}\xspace}
\newcommand{\vpull}{\ensuremath{\textsc{VPULL}}\xspace}
\newcommand{\pushpull}{\push-\pull}
\newcommand{\pushrpull}{\push-\rpull}

\newcommand{\roundpull}{\ensuremath{\textsc{PULL}_1}}
\newcommand{\roundrpull}{\ensuremath{\textsc{RPULL}_{T}}}
\newcommand{\roundvpull}{\ensuremath{\textsc{VPULL}_{T+1}}}

\newcommand{\boldroundpull}{\ensuremath{\mathbf{PULL}_1}}
\newcommand{\boldroundrpull}{\ensuremath{\mathbf{RPULL}_{T}}}
\newcommand{\boldroundvpull}{\ensuremath{\mathbf{VPULL}_{T+1}}}
\newcommand{\boldroundvpullT}{\ensuremath{\mathbf{VPULL}_{T}}}

\newcommand{\boldpull}{\ensuremath{\mathbf{PULL}}}
\newcommand{\boldpush}{\ensuremath{\mathbf{PUSH}}}
\newcommand{\boldrpull}{\ensuremath{\mathbf{RPULL}}}
\newcommand{\boldvpull}{\ensuremath{\mathbf{VPULL}}}

\newcommand{\Srpull}{\ensuremath{S_{T}^{\rpull}}}
\newcommand{\Svpull}{\ensuremath{S_{T+1}^{\vpull}}}
\newcommand{\Spull}{\ensuremath{S_{1}^{\pull}}}
\newcommand{\Srpullz}{\ensuremath{X}}
\newcommand{\Svpullz}{\ensuremath{X^{\textsc{VP}}}}
\newcommand{\Spullz}{\ensuremath{X^\textsc{P}}}

\newcommand{\prp}{\pushpull}
\newcommand{\GoodExecution}{\ensuremath{\mathbb{G}}}

\newcommand{\ddelta}{\ensuremath{\frac{\Delta}{\delta}}}
\newcommand{\ddlogn}{\ensuremath{\frac{\Delta}{\delta}\log n}}
\newcommand\lcb{\left\lbrace} 
\newcommand\rcb{\right\rbrace} 
\newcommand\cb[1]{\lcb #1 \rcb} 
\newcommand\lab{\left[} 
\newcommand\rab{\right]} 
\newcommand\lb{\left(} 
\newcommand\rb{\right)} 
\newcommand\br[1]{\lb #1 \rb} 

\newcommand{\rv}[1]{\ensuremath{\mathrm{#1}}} 
\newcommand{\srv}[1]{\ensuremath{\mathcal{#1}}} 

\newcommand{\card}[1]{\left\vert{#1}\right\vert}

\DeclareMathOperator{\E}{\mathbf{E}}
\let\Pr\relax
\DeclareMathOperator{\Pr}{\mathbf{P}}
\DeclareMathOperator{\powset}{\mathfrak{P}}

\DeclareMathOperator{\polylog}{\ensuremath{\mathrm{polylog}}}
\DeclareMathOperator{\polyloglog}{\ensuremath{\mathrm{polyloglog}}}

\newcommand{\EOf}[1]{\E\mathopen{}\lab #1 \rab\mathclose{}}
\newcommand{\PrOf}[1]{\Pr\mathopen{}\lb #1 \rb\mathclose{}}
\newcommand{\powsetOf}[1]{\powset\mathopen{}\lb #1 \rb\mathclose{}}

\newcommand{\bigO}{\ensuremath{O}}
\newcommand{\bigOOf}[1]{\bigO \mathopen{}\lb #1 \rb\mathclose{}}
\newcommand{\smallO}{\ensuremath{o}}
\newcommand{\smallOOf}[1]{\smallO \mathopen{}\lb #1 \rb\mathclose{}}
\newcommand{\OmegaOf}[1]{\Omega \mathopen{}\lb #1 \rb\mathclose{}}
\newcommand{\omegaOf}[1]{\omega \mathopen{}\lb #1 \rb\mathclose{}}
\newcommand{\ThetaOf}[1]{\Theta \mathopen{}\lb #1 \rb\mathclose{}}

\newcommand{\ID}{\textit{ID}}

\newcommand{\local}{\mathcal{LOCAL}}
\newcommand{\col}{\top}
\newtheorem{theorem}{Theorem}[section]
\newtheorem{lemma}[theorem]{Lemma}
\newtheorem{claim}[theorem]{Claim}
\newtheorem{corollary}[theorem]{Corollary}
\newtheorem{definition}[theorem]{Definition}
\newenvironment{proofsketch}{\begin{proof}[Proof Sketch]}{\end{proof}}

\newcommand{\false}{\ensuremath{\mathbf{false}}}
\newcommand{\true}{\ensuremath{\mathbf{true}}}
\newcommand{\Sstate}{\texttt{informed}\xspace}
\newcommand{\Ustate}{\texttt{uninformed}\xspace}
\newcommand{\state}{\ensuremath{\textit{state}_v}}
\newcommand{\waiting}{\ensuremath{\textit{tokenReceived}_v}}
\newcommand{\ack}{\textit{token}}
\newcommand{\badexec}{\textit{BE}}
\newcommand{\strongly}{\textit{sc}}
\newcommand{\msg}{\textit{msg}}
\newcommand{\textpath}{\textit{path}}

\algnewcommand\algorithmicswitch{\textbf{switch}}
\algnewcommand\algorithmiccase{\textbf{case}}

\algdef{SE}[SWITCH]{Switch}{EndSwitch}[1]{\algorithmicswitch\ #1\ \algorithmicdo}{\algorithmicend\ \algorithmicswitch}%
\algdef{SE}[CASE]{Case}{EndCase}[1]{\algorithmiccase\ #1}{\algorithmicend\ \algorithmiccase}%
\algtext*{EndSwitch} 
\algtext*{EndCase} 
\algdef{SE}[WithProb]{WithProb}{EndWithProb}[1]{\algorithmicwithprob\ #1\ \algorithmicdo}{\algorithmicend\ \algorithmicwithprob}%
\algdef{Ce}[Otherwise]{WithProb}{Otherwise}{EndWithProb}{\algorithmicotherwise}%
\algtext*{EndWithProb}%
\algtext*{EndOtherwise}%

\algnewcommand{\NoAlignComment}[1]{\hspace{1cm}\textit{// #1}}
\algnewcommand{\FullLineComment}[1]{\Statex\textit{// #1}}
\algnewcommand\algorithmicwithprob{\textbf{with probability}}
\algnewcommand\algorithmicotherwise{\textbf{otherwise}}
\algnewcommand\algorithmicdbupdate{\textbf{update}}
\algnewcommand{\DBUpdate}[1]{\State \algorithmicdbupdate\ #1\ \textbf{with: }}
\algnewcommand{\algorithmicgoto}{\textbf{go to}}%
\algnewcommand{\Goto}[1]{\State \algorithmicgoto~\ref{#1}}%
\newcommand{\otherwise}{\textbf{, otherwise }}

\newcommand{\eps}{\epsilon}
\newcommand{\logstar}[1]{\log^{\ast}{#1}}
\newcommand{\poly}[1]{\ensuremath{\mathrm{poly}({#1})}}
\newcommand{\Delt}{\Delta}

\newcommand{\TODO}[1]{ {\color{blue} TODO: #1 }}
\newcommand{\tocheck}[1]{{\color{red} #1}}
\newcommand{\strikeout}[1]{\hide{#1}}
\newcommand{\gone}[1]{{\color{blue}\sout{#1}}}
\newcommand{\added}[1]{{\color{red}#1}}
\newcommand{\hide}[1]{}

\renewcommand{\vec}[1]{\underline{#1}}
\newcommand{\dotcup}{\dcup}
\newcommand{\code}[1]{\texttt{#1}}
\newcommand{\norm}[1]{\lVert #1 \rVert}
\newcommand{\set}[1]{\ensuremath{\left\{#1\right\}}}
\newcommand{\bigbrackets}[1]{\ensuremath{\big(#1\big)}}
\newcommand{\brackets}[1]{\ensuremath{\left(#1\right)}}
\newcommand{\st}{\medspace | \medspace}
\newcommand{\coloneq}{\mathrel{\mathop:}=}

\newcommand{\whpof}[1]{\text{w.h.p.}(\ensuremath{#1})}
\newcommand{\Whpof}[1]{\text{W.h.p.}(\ensuremath{#1})}
\DeclareDocumentCommand \whp { g } {\IfNoValueTF{#1}{\ensuremath{\text{w.h.p.}}\xspace}{\ensuremath{\text{w.h.p.}(#1)}\xspace}}
\DeclareDocumentCommand \Whp { g } {\IfNoValueTF{#1}{\ensuremath{\text{W.h.p.}}\xspace}{\ensuremath{\text{W.h.p.}(#1)}\xspace}}
\newcommand{\wcp}{\text{w.c.p.}\xspace}
\newcommand{\Wcp}{\text{W.c.p.}\xspace}
\newcommand{\wvhp}[1]{\text{w.v.h.p.}(\ensuremath{#1})\xspace}
\newcommand{\Wvhp}[1]{\text{W.v.h.p.}(\ensuremath{#1})\xspace}
\newcommand{\wrt}{\text{with regard to}\xspace}

\newcommand{\Section}[1]{\vspace{-.0in}\section{#1}\vspace{-.0in}}
\newcommand{\Paragraph}[1]{\vspace{.0in}\noindent {{\bf #1}\quad}}
\newcommand{\SectionS}[1]{\vspace{0in}\section*{#1}\vspace{0in}}
\newcommand{\Subsection}[1]{\vspace{-.0in}\subsection{#1}\vspace{-.0in}}
\newcommand{\Subsubsection}[1]{\vspace{0in}\subsubsection{#1}\vspace{0in}}
\newcommand{\paraclose}[1]{\vspace{0.2em}\noindent\textbf{#1.}~}

\newcommand{\FullOrShort}{short}


\ifthenelse{\equal{\FullOrShort}{full}}{
	  
	  \newcommand{\fullOnly}[1]{#1}
	  \newcommand{\shortOnly}[1]{}

  }{
	 \newcommand{\fullOnly}[1]{}
	 \newcommand{\shortOnly}[1]{#1}
  }


\newcommand{\todo}[1]{{\color{red} TODO: #1}}
\newcommand{\remark}[2][]
{{\color{green}\ifthenelse{\equal{#1}{}}{Remark:}{Remark (#1):} #2}}



\date{}
\title{Rumor Spreading with Bounded In-Degree}

\newcommand{\authornamestyle}[1]{#1}
\newcommand{\authordeptstyle}[1]{{\small #1}}
\newcommand{\authorunivstyle}[1]{{\small #1}}
\newcommand{\authorcitystyle}[1]{{\small #1}}
\newcommand{\authormailstyle}[1]{{\small #1}}

\iftrue
\author{
  \authornamestyle{
    Sebastian Daum
  }\\ 
  \authordeptstyle{
    Dept.\ of Comp.\ Science
  }\\
  \authorunivstyle{
  U.\ of Freiburg, Germany
  }\\
  \authorcitystyle{
  Freiburg, Germany
  }\\
  \authormailstyle{
    sdaum@cs.uni-freiburg.de
  }
  \and %
  \authornamestyle{
    Fabian Kuhn
  }\\ 
  \authordeptstyle{
    Dept.\ of Comp.\ Science
  }\\
  \authorunivstyle{
    U.\ of Freiburg, Germany
  }\\
  \authorcitystyle{
    Freiburg, Germany
  }\\
  \authormailstyle{
    kuhn@cs.uni-freiburg.de
  }
  \and %
  \authornamestyle{
    Yannic Maus
  }\\ 
  \authordeptstyle{
    Dept.\ of Comp.\ Science
  }\\
  \authorunivstyle{
    U.\ of Freiburg, Germany
  }\\
  \authorcitystyle{
    Freiburg, Germany
  }\\
  \authormailstyle{
    yannic.maus@cs.uni-freiburg.de
  }
}
\fi

\hide{
\author{ 
\alignauthor
	Sebastian Daum\\
        \affaddr{University of Freiburg}\\
	\affaddr{Freiburg, Germany}\\
	\email{\mbox{sdaum@cs.uni-freiburg.de}}\\
\alignauthor
	Fabian Kuhn\\
        \affaddr{University of Freiburg}\\
	\affaddr{Freiburg, Germany}\\
	\email{kuhn@cs.uni-freiburg.de}\\
\alignauthor
	Christof Sch\"otz\\
	\affaddr{University of Freiburg}\\
	\affaddr{Freiburg, Germany}\\
	\email{schoetzc@informatik.uni-freiburg.de}
}
}




\maketitle

\setcounter{page}{0}
\thispagestyle{empty}

\begin{abstract}
  We consider a variant of the well-studied gossip-based model of
  communication for disseminating information in a network, usually represented by a graph.
  Classically, in each time unit, every node $u$ is allowed to contact
  a single random neighbor $v$. If $u$ knows the data (rumor) to be
  disseminated, node $v$ learns it (known as \emph{push}) and if node $v$
  knows the rumor, $u$ learns it (known as \emph{pull}). While in the
  classic gossip model, each node is only allowed to contact a single
  neighbor in each time unit, each node can possibly be contacted by many
  neighboring nodes. If, for example, several nodes pull from
  the same common neighbor $v$, $v$ manages to inform all these nodes
  in a single time unit.

  In the present paper, we consider a restricted model where at each
  node only one incoming request can be served in one time unit. As long as only a
  single piece of information needs to be disseminated, this does not
  make a difference for push requests. It however has a significant
  effect on pull requests. If several nodes try to pull the
  information from the same common neighbor, only one of the requests
  can be served. In the paper, we therefore concentrate on this weaker
  pull version, which we call \emph{restricted pull}.

  We distinguish two versions of the restricted pull protocol
  depending on whether the request to be served among a set of pull
  requests at a given node is chosen adversarially or uniformly at
  random. As a first result, we prove an exponential separation
  between the two variants. We show that there are instances where if
  an adversary picks the request to be served, the restricted pull
  protocol requires a polynomial number of rounds whereas if the
  winning request is chosen uniformly at random, the restricted pull
  protocol only requires a polylogarithmic number of rounds to inform
  the whole network. Further, as the main technical contribution, we
  show that if the request to be served is chosen randomly, the
  slowdown of using restricted pull versus using the classic pull
  protocol can w.h.p.\ be upper bounded by
  $O(\Delta / \delta \cdot\log n)$, 
  where $\Delta$ and $\delta$ are
  the largest and smallest degree of the network.
\end{abstract}

\section{Introduction}
\label{sec:introduction}
Gossip-based communication models have received a lot of attention as
a simple, fault-tolerant, and in particular also scalable way to
communicate and disseminate information in large networks. The classic
application of gossip-based network protocols is the spreading of
information in the network, specifically the problem of broadcasting a
single piece of information to all nodes of a network, in this context
also often known as \emph{rumor spreading}, e.g.,
\cite{demers87,feige90,frieze85,chierichetti:2010,GIAK11,karp00}. On top of this,
gossip-based protocols have for example also been proposed for
applications such as maintaining consistency in a distributed database
\cite{demers87}, for data aggregation problems
\cite{kempe03,mosk-aoyama08,pandurangan10}, or even to run arbitrary distributed
computations \cite{simulateLOCAL}.

The best studied gossip strategy is the \emph{random phone call
  model}, which was first considered in \cite{frieze85}. We are given
a network graph $G=(V,E)$ where initially a source node $s\in V$ knows
some piece of information (\emph{rumor}) and the objective is to
disseminate the rumor to all nodes of $G$. Typically, time is divided
into synchronized rounds, where in each round, every node can contact
a random neighbor and if $u$ contacts $v$, an interaction between $u$
and $v$ is initiated for the current round. For spreading a rumor, two
basic modes of operation are distinguished. Nodes that already know
the rumor can \push the information to the randomly chosen neighbor
\cite{frieze85} or nodes that do not yet know the rumor can \pull the
information from the randomly chosen neighbor \cite{demers87}. In much
of the classic work, the network $G$ is assumed to be a complete
graph. In that case, it is not hard to see that \push and \pull both
succeed in $\bigO(\log n)$ rounds and that the total number of
interactions of each node can also be bounded by $\bigO(\log n)$. In
\cite{karp00}, it is shown that when combining \push and \pull (in the
following referred to as \pushpull), the average number of
interactions per node is only $\Theta(\log\log n)$.

Mostly in recent years, \push, \pull, and \push-\pull have also been
studied for more general network topologies, e.g.,
\cite{weakconductance,feige90,fountoulakis12,chierichetti:2010,GIAK11,vertexexpansion,giak:2014:vertexexpansion},
with \cite{chierichetti:2010,GIAK11} and \cite{vertexexpansion,giak:2014:vertexexpansion} studying the time complexity as a function 
of the graph's conductance and vertex expansion, respectively.
%
E.g., in
\cite{GIAK11}, it is shown that with high probability (w.h.p.), the
running time of \pushpull can always be upper bounded by
$\bigO((\log n)/\phi(G))$, where $n$ is the number of nodes and $\phi(G)$
is the conductance of the network graph $G$.

While in gossip protocols, each node can initiate at most one
interaction with some neighbor, even if each node contacts a uniformly
random neighbor, the number of interactions a node needs to
participate in each round can be quite large. In complete graphs and more
generally in regular graphs, the total number of interactions per node
and round can easily be upper bounded by $\bigO(\log n)$. However in
general topologies a single node might be contacted by up to
$\Theta(n)$ neighboring nodes. As an extreme case, consider a star
network where a single center node is connected to $n-1$ leaf nodes.
Even if the rumor initially starts at a leaf node, \pushpull manages
to disseminate the rumor to all nodes in only $2$ rounds. Clearly, in
these $2$ rounds, the center node has to interact with all $n-1$ leaf
nodes. In fact, all recent papers which study the time complexity of the random \pushpull
protocol critically rely on the fact that a node can be contacted by many nodes
in a single round, e.g., \cite{GIAK11}.
In some cases, this behavior might limit the implementability and
thus the applicability of the proven results for this gossip protocol. 
In order to obtain scalable systems, ideally, we would like to not
only limit the number of interactions each node initiates, but also
the number of interactions each node participates in.

In the present paper, we therefore study a weaker variant of the
described random gossip algorithms. In each round, every node can
still initiate a connection to one uniformly random neighbor. However,
if a single node receives several connection requests, only one of
these connections is actually established. When disseminating a rumor
by using the \push protocol, this restriction does not limit the
progress of the algorithm. In a given round, a node $v$ learns the
rumor if and only if at least one \push request arrives at
$v$. However, when using the \pull protocol, the restriction can have
a drastic effect. If a node $v$ receives several \pull requests from
several nodes that still need to learn the rumor, only one of these
nodes can actually learn the rumor in the current round. In our paper,
we therefore concentrate on the \pull protocol and we define \rpull
(\emph{restricted} \pull) as the described weak variant of the \pull
algorithm: In each \rpull round, every node that still needs to learn
the rumor contacts a random neighbor. At every node that knows the
rumor, one of the incoming requests (if there are any) is selected and
the rumor is sent to the corresponding neighbor. By \pushrpull we
denote the combination of $\rpull$ with a simultaneous execution of
the classic $\push$ protocol.

\paragraph*{Contributions} 
We first consider two versions of the \rpull protocol which differ in
the way how one of the incoming requests is selected.  Assume that in a
given round some informed node $v$ receives \rpull requests from a set
of neighbors $R_v$. In the \emph{adversarial} \rpull protocol, an
(adaptive) adversary picks some node $u\in R_v$ which will then learn
the rumor. In the \emph{random} \rpull protocol, we assume that a
uniformly random node $u\in R_v$ learns the rumor (chosen
independently for different nodes and rounds). While the choice of
which neighbor a node (actively) contacts with a request is under the control of
the protocol, it is not necessarily clear how one of the incoming
requests in $R_v$ is chosen. If the node can only answer one request
per time unit and the requests do not arrive at exactly the same time,
the first request might be served and all others dropped. Or even if
requests arrive at the same time, it might be the underlying network
infrastructure or operating system which picks one request and drops
the others. If it is reasonable to assume that the incoming requests
are served probabilistically and independently, we believe that random
\rpull provides a good model. Otherwise, the adversarial assumption allows to study the worst-case behavior.

As a first result, we show that there
are instances for which there is an exponential gap between the
running times of the two \rpull variants. We give an instance where for every source node
the random \rpull protocol informs all nodes of the network in
polylogarithmic time, w.h.p., whereas, for every source, the adversarial \rpull algorithm
requires time $\tilde{\Omega}(\sqrt{n})$ to even succeed with a
constant probability.

In the second part of the paper, we have a closer look at the
performance of the random \rpull protocol. Consider a graph $G$ and
let $\delta$ and $\Delta$ denote the smallest and largest degree of
$G$. In each round, in expectation, each informed node receives at
most $\Delta/\delta$ requests. Hence, if an uninformed node $u$
sends an \rpull request to an informed node, $u$ should receive the
rumor with probability at least $\Omega(\delta/\Delta)$. Consequently,
intuitively, the slowdown of using random \rpull instead of the usual
\pull protocol should not be more than  $\tilde{\bigO}(\Delta/\delta)$. \footnote{Here $\tilde{\bigO}$ hides $\log(n)$ factors.} 
We prove that this intuition is correct. For every given instance, we
show that if the \pull algorithm informs all nodes in $\TTT$ rounds with
probability $p$, for the same instance, the random \rpull algorithm
manages to reach all nodes in time $\bigO\big(\TTT\cdot\frac{\Delta}{\delta}\cdot
\log n\big)$ with probability $(1-o(1))p$.%
\footnote{Actually, $\frac{\Delta}{\delta}$ can be replaced by $\max_{\set{u,v} \in E} \deg(u)/\deg(v)$ in all parts of the paper.}
While the statement might seem
very intuitive, its formal proof turns out quite involved. Formally,
we prove a stronger statement and show that a single round of the
\pull protocol is w.h.p.\ \emph{stochastically dominated} by
$\bigO\big(\frac{\Delta}{\delta}\cdot \log n\big)$ rounds of random \rpull in the
following sense. We give a coupling between the random processes
defined by \pull and random \rpull such that for every start
configuration, w.h.p., the set of nodes informed after
$\bigO\big(\frac{\Delta}{\delta}\cdot\log n\big)$ rounds of random \rpull is a
superset of the set of nodes informed in a single \pull round. The same holds for simulating one round of \pushpull with \pushrpull. 
A similar coupling between rumor spreading algorithms has been done in \cite{AcanMehrabianWorlmald14} where the authors couple $\log(n)$ rounds of asynchronous-  with one round of synchronous \pushpull. A coupling between \pull and \rpull in the classic sense, i.e., a coupling which does relinquish the \whp term does not exist. We also
show that for such a round-by-round analysis, our bound is tight. That
is, there are configurations where $\Omega\big(\frac{\Delta}{\delta}\log
n\big)$ random \rpull rounds are needed to dominate a single \pull
round with high probability. 

\paragraph*{Notation and Preliminaries}
Let $G=(V,E)$ be the $n$-node network graph. For a node $u\in V$, we
use $N(u)$ to denote the set of neighbors of $u$ and $d_u=d(u):=|N(u)|$ to
denote its degree. Given a set of nodes $S\subseteq V$, we
define $N_S(u):=N(u)\cap S$ to be the set of $u$'s neighbors in $S$
and $d_S(u):=|N_S(u)|$ for the number of neighbors of $u$ in $S$. The
smallest and largest degrees of $G$ are denoted by $\delta$ and
$\Delta$, respectively. For a set $V' \subseteq V$ we denote with
$G[V']$ the graph induced by $V'$. To indicate a disjoint union of two sets, i.e., $A\cup B$ with $A \cap B=\emptyset$, we write $A \dcup B$. For a set of natural numbers $\set{1, \dots, k}$ we only write $[k]$.

When analyzing the progress of an algorithm, we typically use $S$ to
be the set of initially informed nodes and $U$ to be the initially
uninformed nodes. Given some algorithm \alg, the set $S_t^\alg$
denotes the set of informed nodes after $t$ rounds of \alg when
starting with the set $S_0^\alg := S$ of informed nodes.

\section{Separation of Adversarial and Random \boldmath\rpull}
\label{sec:separation}

We want to show that the adversarial \rpull can be exponentially slower than the randomized \rpull on general graphs. To show this, we first establish results on the run time of both algorithms on trees. These results might also be of independent interest.

In a tree network let $p_{v,u}=(v=v_0,v_1,v_2,\dots,v_q=u)$ denote the unique path from $v$ to $u$, though we use that notation also for the set of nodes on that path, i.e., $p_{v,u}=\set{v,v_1,v_2,\dots,v_{q-1},u}$. Define $D_p := \sum_{w \in p} d_w$, i.e., the sum of all degrees on the path $p$.



The next lemma shows that on a tree any form of \rpull is asymptotically as fast as \pull plus an additive term in the order of the degree of the node that initially has the rumor.

\begin{lemma}\label{thm:treepath}
Let $G$ be a tree network with $S_0=\set{r}$ and let $u$ be a node in $U_0$.
Furthermore, let $\tau$ be the first round in which $u \in S_\tau$ holds, i.e., the number of rounds until $u$ gets informed.

\begin{myenumerate}[(1)]
  \item \label{thm:treepath:pull} $\EOf{\tau}=\Theta(D_{p_{r,u}}-d_r)$ for  \pull, 
  \item \label{thm:treepath:rpull} $\EOf{\tau}=\Omega(D_{p_{r,u}}-d_r)$ for every type of \rpull, 
  \item \label{thm:treepath:adversary} $\EOf{\tau}=\bigO(D_{p_{r,u}})$ for adversarial \rpull.
\end{myenumerate}   
\end{lemma}

\begin{proof}[Proof of Lemma \ref{thm:treepath}]
(\ref{thm:treepath:pull}) We root the tree at the only informed node, $r$. Note that nodes are not aware of their own parent/child relationships.
Consider some time $t$ at which a node $r'$ on the path $p_{r,u}$ is in $S_t\setminus S_{t-1}$, i.e., it just got informed. Thus its child $u' \in p_{r,u}$ on the path is not yet informed, i.e., $u' \in U_t$. In any round $t' \geq t$, in which $u'$ is not informed yet, it requests its parent with probability $1/d_{u'}$. Thus each uninformed node $u' \in p_{r,u}\setminus\set{r}$ on the path needs $\Theta(d_{u'})$ rounds in expectation before it can get informed. Linearity of expectation proves the claim for \pull. 

(\ref{thm:treepath:rpull}) follows from the fact that \rpull is at most as fast as \pull.

(\ref{thm:treepath:adversary}) For adversarial \rpull divide all rounds $t' \geq t$ in which $u'$ is not yet informed into two types: First rounds in which at least one \emph{sibling} of $u'$, i.e., the nodes in $N(r')\backslash \{u'\}$, requests from $r'$ and secondly rounds in which no sibling of $u'$ requests from $r'$. The first type of rounds is upper bounded by $d_{r'}$ because every neighbor of $r'$ stops requesting after receiving the rumor. In expectation $u'$ gets the rumor after $d_{u'}$ rounds of type two; thus in expectation $u'$ is informed within $\bigO(d_{r'}+d_{u'})$ rounds.  Applying this recursively to all uninformed nodes on the path $p_{r,u}$, we get the claimed result via linearity of expectation.
\end{proof}

\begin{lemma}\label{lemma:treefull}
  Let $G$ be a tree network with $S_0=\set{r}$. Then in both random and adversarial \rpull it takes $\bigO\big(\max_{\textpath \, p} D_p + \Delta \log n\big)$ rounds to fully inform all nodes in $V$, \whp.
\end{lemma}
\begin{proof}[Proof of Lemma \ref{lemma:treefull}]
  
  The analysis we use does hold for adversarial \rpull. 
  
  First look at a path $p=p_{r,l}=(r=v_0,v_1,v_2,\dots,v_q=l)$ from the root to some leaf $l$ and define $\Delta_p = \max_{v \in p} d_v$.
  Let $T_i$ be the random variable that indicates the round in which node $v_i$ gets informed and with $T_0:=0$ we can define $X_i := T_i-T_{i-1}$ for $m=1,\dots,q$, the time node $v_{i-1}$ needs to pass the information forward to node $v_i$. For simplicity in the following we define $u:=v_i$ and $s:=v_{i-1}$.
	
	
  Once $s$ gets informed, a round is called free if no node in $N(s)\setminus\{v_{i-2},u\}$ requests to $s$, otherwise it is called congested. In a free round, $u$ gets the rumor with probability at least $1/d_u$, i.e., the number of free rounds is upper bounded by a geometric random variable $Y_{i,\text{free}} \sim Geom(1/d_u)$. On the other hand, in a congested round, at least one neighbor of $s$ does get informed, so there can not be more than $d_s$ such rounds. 


In total we get that $\Pr(X_i \geq \tau) \leq \Pr(Y_{i,\text{free}} + d_s \geq \tau)$ for node $u=v_i$ and with $T_q=\sum_{i=1}^q X_i$ we conclude that $\Pr(T_q \geq \tau) \leq \Pr(\sum_{i=1}^q Y_{i,\text{free}}+\sum_{i=1}^q d_{v_{i-1}} \geq \tau)$.

Since we are on a tree, those geometric random variables are all independent, and we can apply the Chernoff Lemma \ref{thm:geometricchernoff}. Let $\Delta_p$ be the largest degree of all nodes on the path excluding $r$, i.e., $p_1 = 1/\Delta_p$ in terms of the notation from Lemma \ref{thm:geometricchernoff}. With $Y=\sum_{i=1}^q Y_{i,\text{free}}$ we have $\mu = D_p - d_r \geq \Delta_p$ and we set $t = c \Delta_p \log n$ for some $c>0$. 
$$
\Pr(Y > 3(\mu+t)) \leq e^{-\frac{\mu}{2\Delta_p}} e^{-\frac{c \Delta_p \log n}{\Delta_p}} < n^{-c},
$$
i.e., \whp, $v_q$ is informed within $\bigO(D_p + \Delta_p \log n) = \bigO(\max_{\textpath \, p} D_p + \Delta \log n)$ rounds. In a tree there are at most $n$ root-leaf paths, therefore a union bound over all individual paths concludes the proof.
\end{proof}

Lemma \ref{lemma:treefull} shows that random \rpull and adversarial \rpull are essentially the same on trees. This does not hold for general graphs. 

\begin{lemma}
\label{lemma:separation}
  There is a graph $G=(V, E)$ of size $\Theta(n)$ with node $r_\alpha \in V$, $\deg(r_\alpha)\leq 3$, such 
  that:
  \begin{myitemize}
    \item 
    For $S_0=\set{r_\alpha}$, \wcp, the run-time of adversarial \rpull is in $\Omega(\sqrt{n})$.
    \item For any non-empty $S_0 \subset V$, \whp, the run-time of randomized \rpull is in $\bigO(\log^2 n)$.
  \end{myitemize}
\end{lemma}

\iftrue
\begin{proof}[Proof of Lemma \ref{lemma:separation}]
  A picture depicting the graph can be found in Appendix \ref{app:pictures}.

  We use the notation \whp{k} to say that some event holds with high probability with respect to $k$, i.e., with probability at least $1-1/k$.
  
  We first introduce a graph type, with size depending on some parameter $k$, that we call a $k$-\dbtfull (\dbt{k}). In simple words, a $\dbt{k}$ is a binary tree with $k$ leaves, but with those $k$ leaves being fully interconnected, i.e., forming a clique. Mathematically more precise, a graph \mbox{$H=(B \dotcup L, E_B \dotcup E_L)$} with $|L|=k$ and $|B|=k-1$ is a $\dbt{k}$, iff $H_L = (L, E_L)$ is a complete graph over $k$ nodes and $H_B=(B\dotcup L, E_B)$ is a complete binary tree with its leaves being the nodes in $L$. While $H$ is not strictly speaking a tree, we call nodes in $L$ its leaves, $L(H)=L$ the \emph{leaf set}, $B(H)=B$ its \emph{branch set} and the root of $H_B$ we call the \emph{root} of $H$. Also, every node in $H$ except for its root has a clearly designated \emph{parent} (defined by $H_B$) and each node in $B$ has two clearly designated \emph{children} (with respect to the root in $H_B$).

  \begin{claim}
    \label{claim:dbt}
    Let $G=(V,E)$ be a graph and $H=(V_H, E_H)$ be a subgraph of $G$ that is a $\dbt{k}$. Furthermore, let any node $v \in V_H$ have at most one connection outside of $V_H$, i.e., $\deg_G(v) \leq \deg_{H}(v) + 1$. Then, \whp{k}, for any non-empty set of nodes in $H$ knowing the rumor, randomized \rpull informs all nodes in $H$ within $\bigO(\log k)$ rounds.
  \end{claim}
  \begin{quotation}
    Without loss of generality let there be one node $s$ having the rumor. If $s \in B:=B(H)$, i.e., $\deg_G(s) \leq \deg_H(s) +1 \leq 4$, then we can apply \Cref{lemma:treefull} to get that \emph{all} nodes in $B$ are informed within $\bigO(\log k)$ rounds. Let this be the case. All nodes in $L:=L(H)$ have degree at most $k+1$: $k-1$ neighbors in $L$, one ``parent-node'' in $B$ and at most one neighbor in $V \setminus V_H$. Each of them requests to its neighboring parent from $B$ with probability at least $1/(k+1)$, i.e., in each round, with probability at most $\big(1-\frac{1}{k+1}\big)^k < 1/2$, no node in $L$ learns the rumor. By Chernoff, \whp{k}, after $\bigO(\log k)$ rounds, at least one node in $L$ knows the rumor. If $x>0$ nodes in $L$ are informed, then each uninformed node $u$ in $L$ requests from one of those $x$ nodes (or a node in $B$) with probability at least $x/(k+1)$ and with probability at least $(1-1/k)^{k-1} > 1/3$ node $u$ is the only node requesting from its target. As long as $x < k/2$, with linearity of expectation, each round the expected number of newly informed nodes in $L$ is in $\Omega(x)$. Once $x \geq k/2$ we can use a similar argument to show that, \wcp, the number of uninformed nodes goes down by a constant factor each round. Hence, after $\bigO(\log k)$ rounds in expectation, but also \whp{k}, all nodes in $L$ are informed.%

    If initially $s \in L$, then with probability $p \in \set{1/4,1/3}$ its parent node in $B$ requests from $s$, while at the same time with probability at least $(1-1/k)^{k-1} > 1/3$ no other node in $L$ requests from $s$. Hence, after $\bigO(\log k)$ rounds, \whp, the parent node gets the rumor from $s$. The rest follows from reduction to the first case.
    \qed
  \end{quotation}
  
  We construct $G=(V_1 \dotcup V_\zeta, E)$ as follows. We let $D_\alpha$ and $D_\zeta$ be two $\dbt{n}$s, and we have $m$ $\dbt{\mtwo}$s that we denote with $D_1, D_2, \dots, D_{m}$, 
where $\mtwo:=\sqrt{n}$ and $m:=c\sqrt{n}$ for some natural number $c$. We use the notation $D_i$ for the corresponding \dbt{k} and its node set interchangeably. 
Their corresponding roots and leaf sets are denoted as $r_\alpha, r_\zeta, r_1, r_2, \dots r_{m}$ and $L_\alpha, L_\zeta, L_1, L_2, \dots L_{m}$ respectively, and with $l_{X,1}, l_{X,2}, \dots$ we enumerate the leaves of leaf set $L_X$. Let $C_\alpha=\set{c_1, \dots c_{m}}$ be an arbitrary $m$-sized subset of $D_\alpha$'s branch set $B_\alpha$ -- for simplicity and in accordance to \Cref{fig:separation} think of $C_\alpha$ as the layer of nodes in $B_\alpha$ that are at depth $\log m$. 

  We let $V_\zeta=D_{\zeta}$ and $V_1=\{r\}\cup D_{\alpha}\cup D_1\cup \ldots\cup D_m$ and we add the following edges.
	 \TabPositions{5.0cm}
  \begin{myitemize}
 \item Between $r$ and $D_{\zeta}$:\tab We add one edge from $r$ to $r_\zeta$.
	\item Between $r$ and $D_{\alpha}$:\tab For each $j \in [m\log n]$ we add an edge from $r$ to $l_{\alpha,j}$.
\item Between $r$ and $D_1,\ldots,D_m$:\tab For each $i \in [m]$ and $j \in [\log n]$ we add one edge from $r$ to $l_{i,j}$.
	\item Between $D_1,\ldots,D_m$ and $D_{\zeta}$:\tab For each $i \in [m]$ we add one edge from $r_i$ to $l_{\zeta,i}$.
   \item Between $D_1,\ldots,D_m$ and $C_{\alpha}$:\tab For each $i \in [m]$ we add one edge from $l_{i,\mtwo}$ to $c_i$.
  \end{myitemize}
  Note that the degree $\deg(r)$ is $2m\log n +1$ and that all the above defined edges add to any node in a \dbt at most one edge that connects it to a node outside its own \dbt.

   The idea of the proof is the following: The graph is built in a way that information propagation from $V_\zeta$ to $V_1$ is quick, but not the other way round.
  In the random \rpull model, wherever the rumor starts, it reaches $r$ quickly and from there $r_\zeta$ manages to get the rumor from $r$ in polylogarithmic time. Then the rumor  quickly  propagates through $V_{\zeta}=D_\zeta$, and from $L_\zeta$ to all \dbts $D_1, \dots, D_{m}$ and afterwards to $D_{\alpha}$.
	
	In the adversarial \rpull model, as long as the rumor does not start in $V_{\zeta}$, the rumor can quickly spread to $r$, a few of the $D_i$s and $D_\alpha$ but not to $V_{\zeta}$ because we let the adversary always prioritize a request at node $r$ from a node in one of the $D_i$s over a request from $r_\zeta$ to prevent that $r_{\zeta}$ will get the rumor. This is possible because we show that for polynomially many rounds there is always a request at $r$  from one of the $D_i$s to serve. Thus, to inform $D_{\zeta}$ all information must go through one of the edges $\set{r_i,l_{\zeta,i}}$, $i=1,\ldots,m$,  with $r_i$ informed. In less than a polynomial number of rounds few enough of the $r_i$s are informed and in each round only few requests from the leaf nodes $L_{\zeta}$ request from one of the $r_i$s at all making it unlikely that one of them requests from an informed $r_i$. Hence propagation through one of these edges is unlikely and it takes a long time for the rumor to spread over the entire graph.


  \paraclose{Random \rpull}
  We start proving that random \rpull manages to spread the rumor quickly in $G$. 
  \begin{myenumerate}[(1)]
    \item If there is an informed node in $D_\zeta$, by \Cref{claim:dbt}, \whp, all of $D_\zeta$ is informed in $\bigO(\log n)$ rounds. Assume this has happened. Since each root of a \dbt $D_i$ has degree $3$ in $G$, it requests the rumor from an informed leaf node in $L_\zeta$ \wcp -- since no other node in $L_\zeta$ is still uninformed and therefore able to create a conflict, \whp, in $\bigO(\log n)$ rounds, \emph{all} root nodes $r_1, \dots, r_m$ know the rumor.
    \item If there is an informed node in $D_i$ for some $i=1,\ldots,m$, due to \Cref{claim:dbt}, the whole \dbt $D_i$ is informed \whp{m} ($=$\whp) within $\bigO(\log m)=\bigO(\log n)$ rounds. Assume this has happened. Node $c_i \in C_\alpha$ has degree at most $4$ and therefore requests from its neighboring node $l_{i,m}$ \wcp, and since all nodes in $D_i$ are informed, it will also get the rumor.
    \item If there is an informed node in $D_\alpha$, by \Cref{claim:dbt}, \whp, all of $D_\alpha$ is informed in $\bigO(\log n)$ rounds. Assume this has happened. Almost half of all neighbors of $r$ lie in $L_\alpha$, and with same reasoning as above, $r$ gets the rumor \whp within $\bigO(\log n)$ rounds.
    \item Let $r$ be informed. All its neighbors in $L_\alpha$ have degree $n+1$ and therefore request with probability at most $1/n$ from $r$, i.e., in expectation no more than $1$ node from there requests the rumor from $r$ each turn. Each neighboring leaf node in some $L_i$ has degree $m+1$, i.e., requests the rumor from $r$ with probability at most $1/m$. Since $r$ has $m\log n$ such neighbors, in expectation no more than $\log n$ such neighbors request from $r$. With a Chernoff bound, \whp there are no more than $\bigO(\log n)$ requests at $r$. Since $r_\zeta$ has degree $3$, it therefore requests \wcp and gets the rumor with probability $\Omega(1/\log n)$. \Whp, the rumor is therefore propagated to $r_\zeta$ in $\bigO(\log^2 n)$ rounds.
  \end{myenumerate}
Altogether, wherever the source node is located, the above reasoning shows that, \whp, the rumor is propagated to all nodes within $\bigO(\log^2 n)$ rounds.

  \paraclose{Adversarial \rpull}
  Let $s \in V$ be the source node with the rumor. 
  If $s \in D_i$ for some $i \in [m]$, then, without loss of generality, we assume that all nodes in $D_i$, $D_\alpha$ and $r$ are already informed, initially. Otherwise we inform all nodes in $D_\alpha$ and $r$.
  For $i \in [m]$ we call any $D_i$ \emph{informed}, if it contains at least one informed node, otherwise \emph{uninformed}.

  The adversary has the following simple strategy. If $r_\zeta$ and at least one other node requests the rumor from $r$, then $r$ chooses to pass the rumor to any other node than $r_\zeta$. In every other aspect it follows an arbitrary strategy. 

For time $t$ we denote with $X_t$ the number of informed \dbts $D_i$, and we assume without loss of generality that the corresponding \dbts are $D_1, \dots, D_{X_t}$. Let $\eEvent_t$ be the event that in round $t$ no node in $D_\zeta$ has the 
rumor.
Conditioning on this event implies that, by the structure of our graph and our model, nodes from \dbts $D_i$ need to get the rumor from either $r$ or from $D_\alpha$, via connections $\set{c_i,l_{i,m}}$.

  Let $\xEvent_t$ be the event that $X_t < 4ct+2\log n$, $\aEvent_t$ the event that $r_\zeta$ gets the rumor in round $t$ and let $\cEvent_t$ be the event that a node from $L_\zeta$ gets the rumor from one of the roots $r_i$. 


  \begin{claim}
    $\Pr(\xEvent_t|\eEvent_t) \geq 1-1/n \geq e^{-1/n}$ for any $t \leq \mtwo/\cnenner$.
  \end{claim}
  \begin{quotation}
    In each round, $r$ can inform at most one node in a yet uninformed \dbt $D_i$.
    Also, any uninformed node $l_{i,\mtwo}$ connects to its neighbor $c_i$ in $D_\alpha$ only with probability $1/(\mtwo+1)<1/\mtwo$. With at most $m$ such uninformed nodes trying to get the rumor from $D_\alpha$ each round, the amount of nodes $l_{i,\mtwo}, i=1,\ldots,m$ informed through such an edge is upper bounded by a Binomial random variable $\text{Bin}(tc\mtwo,1/\mtwo)$. Let $X'_t$ be the random variable that counts the number of times when an uninformed \dbt $D_i$ gets informed through such an edge to $D_\alpha$ but not through a connection to $r$. 
Then, by Chernoff, for $\delta = (3-\frac{1}{c})+ 2\frac{\log n}{ct}> 1$,
    $$
    \Pr(X_t \geq t+ ct+\delta ct|\eEvent_t) \leq \Pr(X'_t \geq (1+\delta)ct) \leq \exp\left(- \frac{tc}{2}\delta\ln (1+\delta)\right) 
    \leq
    n^{-1}.
    $$
    Therefore, \whp, $X_t$ is smaller than $4ct+2\log n$ for $t\leq \mtwo/\cnenner$. \qed
  \end{quotation}

  \begin{claim}
    $\Pr(\aEvent_{t+1}|\eEvent_t \cap \xEvent_t) \leq 1/n \leq 1-e^{-1/n}$ for any $t \leq \mtwo/\cnenner$.
  \end{claim}

  \begin{quotation}
    Every uninformed node $l_{i,j} \in L_i$, where $j\in[\log n]$ and $i \in [m]$, requests from $r$ with probability $1/(\mtwo+1)$. 
		Choosing $c$ large enough $\xEvent_t$ implies that at least $m/2$ of the $D_i$s are uninformed. Thus there are at least $(m/2)\cdot\log n$ uninformed leaf nodes with a connection to $r$ in uninformed \dbts $D_i$. At least $0.4c\log n$ such nodes request from $r$ in expectation. Choosing $c$ large enough, a simple Chernoff bound gives us that, \whp, at least one of these nodes requests from $r$. Consequently, \whp, $r$ does not give the rumor to $r_\zeta$ in round $t+1$. \qed
  \end{quotation}
\clearpage
  \begin{claim}
    $\Pr(\cEvent_{t+1}|\eEvent_t \cap \xEvent_t) \leq 1-e^{-\frac{1}{n}(2ct+\log n)}$ for any $t \leq \mtwo/\cnenner$.
  \end{claim}

  \begin{quotation}
    By our assumption of $\eEvent_t$, at the start of round $t+1$, no node in $D_\zeta$ has the rumor, so for $\cEvent_{t+1}$ to possibly happen, a node from $L_\zeta$ must request from one of the nodes $r_1, \dots, r_{X_t}$, which it does with probability $1/(n+1)<1/n$.
    The probability for $\cEvent_{t+1}$ to happen is therefore
    \begin{equation*}
      \Pr(\cEvent_{t+1}|\eEvent_t \cap \xEvent_t) 
      \leq
      1-\left(1-\frac{1}{n}\right)^{X_t}
      \leq 
      1-\exp\left(-\frac{1}{n}(2ct+\log n)\right). \qedhere
    \end{equation*}
  \end{quotation}

  We know that $\overline{\aEvent_{t+1} \cup \cEvent_{t+1}} \cap \eEvent_t \cap \xEvent_t \subseteq \eEvent_{t+1}$ since if under condition $\eEvent_t$ neither $\aEvent_{t+1}$ nor $\cEvent_{t+1}$ happens, then no node in $D_\zeta$ can get informed in round $t+1$.

  \begin{claim}
    $\Pr(\eEvent_t) \geq e^{-\frac{1}{n}(2ct^2+t\log n)}$ for any $t \leq \mtwo/\cnenner$.
  \end{claim}
  \begin{quotation}
    The proof follows by induction. In round $t=0$ clearly no node in $D_\zeta$ is informed, so the induction base holds. For the following, note that conditioned on $\eEvent_t$, events $\aEvent_{t+1}$ and $\cEvent_{t+1}$ (and therefore also their complements) are independent.
    \hide{
    \begin{eqnarray*}
      \Pr(\eEvent_{t+1}) 
      &\geq& \Pr(\overline{\aEvent_{t+1} \cup \cEvent_{t+1}} \cap \eEvent_t \cap \xEvent_t) 
      = \Pr(\eEvent_t \cap \xEvent_t) \Pr(\overline{\aEvent_{t+1}} \cap \overline{\cEvent_{t+1}}|\eEvent_t \cap \xEvent_t) 
      \\
      &\geq& 
       \Pr(\eEvent_t) \Pr(\xEvent_t|\eEvent_t) \Pr(\overline{\aEvent_{t+1}}|\eEvent_t \cap \xEvent_t) \Pr(\overline{\cEvent_{t+1}}|\eEvent_t \cap \xEvent_t)
      \\
      &\geq& 
      \exp\left(-\left(\frac{1}{n}(2ct^2+t\log n)+\frac{1}{n}+\frac{1}{n}+\frac{1}{n}(2ct+\log n)\right)\right)
      \\
      &\geq&
      \exp\left(-\frac{1}{n}\left(2c(t+1)^2+(t+1)\log n\right)\right). \qed
    \end{eqnarray*}
  }
  
  \begin{align*}
    \pushQED{\qed}
    \Pr(\eEvent_{t+1}) 
    &\geq \Pr(\overline{\aEvent_{t+1} \cup \cEvent_{t+1}} \cap \eEvent_t \cap \xEvent_t) 
    = \Pr(\eEvent_t \cap \xEvent_t) \Pr(\overline{\aEvent_{t+1}} \cap \overline{\cEvent_{t+1}}|\eEvent_t \cap \xEvent_t) 
    \\
    &\geq 
    \Pr(\eEvent_t) \Pr(\xEvent_t|\eEvent_t) \Pr(\overline{\aEvent_{t+1}}|\eEvent_t \cap \xEvent_t) \Pr(\overline{\cEvent_{t+1}}|\eEvent_t \cap \xEvent_t)
    \\
    &\geq 
    \exp\left(-\left(\frac{1}{n}(2ct^2+t\log n)+\frac{1}{n}+\frac{1}{n}+\frac{1}{n}(2ct+\log n)\right)\right)
    \\
    &\geq
    \exp\left(-\frac{1}{n}\left(2c(t+1)^2+(t+1)\log n\right)\right). \qedhere
    \popQED
  \end{align*}

  \end{quotation}
  This means, that after $t=\sqrt{n/c}$ rounds with probability at least $e^{-3}$ still not all $\Theta(n)$ nodes in $G$ are informed, concluding the proof of \Cref{lemma:separation}.
\qed
\end{proof}

\fi

\begin{theorem}
  \label{thm:separation}
    There is a graph $G=(V, E)$ of size $\Theta(n)$, such that for any $S_0 = \set{s} \subset V$: 
    \begin{myitemize}
        \item 
        In
        expectation, the run-time of
        adversarial \rpull is in $\Omega(\sqrt{n})$.
        \item 
        \Whp,
        the run-time of randomized
        \rpull is in $\bigO(\log^2 n)$.
    \end{myitemize}
\end{theorem}

\begin{proof}
  Let $G'$ and $G''$ be duplicates of the graph $G$ from \Cref{lemma:separation}, $r'_\alpha$ and $r''_\alpha$ being the respective duplicates of $r_\alpha$. We set $G \coloneqq G' \cup G''$ and add the edge $\set{r'_\alpha,r''_\alpha}$. Without loss of generality let $s \in V'$. 

In the random version, the rumor propagates through all of $G'$ in $\bigO(\log^2 n)$ rounds. Due to its low degree, $r''_\alpha$ gets the rumor from $r'_\alpha$ within $\bigO(\log n)$ time after $G'$ is informed and again, in $\bigO(\log^2 n)$ rounds $G''$ is informed completely. 

In the adversarial version, $G''$ can only learn the rumor from $G'$ through edge $\set{r'_\alpha,r''_\alpha}$. But once $r''_\alpha$ knows the rumor, we can apply \Cref{lemma:separation} again to prove that now progress is stalled.
\end{proof}

\section{Comparison of \boldmath$\pull$ and \boldmath$\rpull$}
\label{sec:pullvsrpull}

In this section we compare the two algorithms \pull and random \rpull on general graphs, i.e., we analyze how many rounds of random \rpull are enough to cover the progress of one round of \pull. More precisely, we show that \whp the set of nodes informed after
$\bigO\big(\frac{\Delta}{\delta}\cdot\log n\big)$ rounds of random \rpull is a
superset of the set of nodes informed in a single \pull round. We manage to do so by coupling both algorithms. At the end of the section we head out to prove that this bound is tight. Whenever we talk about \rpull in this section we mean random \rpull.
\subsection{Dominance and Couplings}

We begin with two examples of insufficient definitions of domination between two rumor spreading algorithms.

Showing for two algorithms $\mathcal{A}$ and $\mathcal{A'}$ that $\Pr\big(u \in S^{\mathcal{A}}\big) \geq \Pr\big(u \in S^{\mathcal{A'}}\big)$ holds for all $u\in U$ is not enough to obtain a natural dominance definition of $\mathcal{A}$ over $\mathcal{A'}$, since due to dependencies for a set $M$ with $|M|>1$ it might still be true that $\Pr\big(M \subseteq S^{\mathcal{A}}\big) < \Pr\big(M \subseteq S^{\mathcal{A'}}\big)$.

Showing that $\Pr\big(M \subseteq S^{\mathcal{A}}\big) \geq \Pr\big(M \subseteq S^{\mathcal{A'}}\big)$ (*) holds for all $M\subseteq U$ is not enough either. Assume the following example: Let $U=\set{a,b,c}$ be the set of uninformed nodes. Assume that under $\mathcal{A}$ the probability that the set of newly informed nodes equals $\{a,b,c\}$, $\set{a}$, $\set{b}$ or $\set{c}$ is $1/8+\eps$ each and the probability that it equals one of the sets $\set{a,b}$, $\set{a,c}$, $\set{b,c}$ or $\emptyset$ is $1/8-\eps$ each. Under $\mathcal{A'}$ we inform any of those sets with probability $1/8$. A direct computation for all $M\subseteq \{a,b,c\}$, e.g., for $M=\{a\}$, $\Pr(\{a\}\subseteq S^{\mathcal{A'}})=1/2$ and 
\[\Pr(\{a\}\subseteq S^{\mathcal{A}})=\Pr(S^{\mathcal{A}}=\{a\})+\Pr(S^{\mathcal{A}}=\{a,b\})+\Pr(S^{\mathcal{A}}=\{a,c\})+\Pr(S^{\mathcal{A}}=\{a,b,c\})=1/2,\]
  shows that inequality (*) is fulfilled for any $M\subseteq U$, but the probability of the event \emph{``at least $2$ nodes are informed''} is by $2\eps$ smaller for $\mathcal{A}$ than for $\mathcal{A'}$.
Hence, to cover most possibly arising cases, we use the definition of the so called (first order) stochastic dominance.

\subsubsection*{Stochastic Dominance and Coupling}
Let $\left(\SSS, \preceq_{\SSS}\right)$ be a finite distributed lattice and let $X_1$ and $X_2$ be random variables with distributions $\Pr_1$ and $\Pr_2$ which take values in $\SSS$. A function $f:\SSS\rightarrow \RR$ is called \emph{increasing} if $A\preceq_{\SSS} B$ implies $f(A)\leq f(B)$.
\hide{The upper closure $\cl\left(\theW\right)$ of a set $\theW\subseteq \SSS$ is defined as 
\begin{align}
\label{closure}
\cl\left(\theW\right)=\{W \in \SSS \mid \exists X\in \theW, X \preceq_{\SSS} W\}.
\end{align}
}
\begin{definition}[Stochastic Dominance]
  \label{def:stochasticDomination}
  We say that $X_1$ \emph{stochastically dominates} $X_2$ if 
  \[E(f(X_1))\geq E(f(X_2))\] holds for \emph{every} increasing function $f:\SSS\rightarrow \RR$, where $E(\cdot)$ denotes the expected value.
\end{definition}
In this paper we will set $\SSS=2^{U}$ to be the power set of $U$, where $U\subseteq V$ is the set of uninformed nodes, $\preceq_{\SSS}$ equals the subset relation on $U$ and $X_1$ and $X_2$ will be the respective random variables describing which nodes get informed in \pull and \rpull. 

Alternative to the definition, one can show that one process stochastically dominates a second process by defining a monotone coupling between the processes (cmp. Theorem \ref{thm:Strassen}).

\begin{definition}[(Monotone) Coupling]
  A coupling of two random processes $X_1$ and $X_2$, taking values in $\SSS$ with distributions $\Pr_1$ and $\Pr_2$, is a joint distribution $\hat{\Pr}$ of a random process $(\hat{X_1},\hat{X_2})$ taking values in $\SSS\times\SSS$, such that its margins stochastically equal the distributions of $X_1$ and $X_2$ respectively, i.e., 
  \label{def:monotonCoupling}
      \begin{align*}
			\sum_{B\in\SSS}\hat{\Pr}\left((\hat{X_1},\hat{X_2}) =(A,B)\right)& =\Pr_1(X_1=A)~\forall
      A\in \SSS\text{ and}\\
     \sum_{A\in\SSS}\hat{\Pr}\left((\hat{X_1},\hat{X_2}) =(A,B)\right) & =\Pr_2(X_2=B)~\forall
      B\in\SSS.
			\end{align*}
  A coupling is called monotone (written $X_1\leq X_2$) if additionally the following holds:
  \begin{align}
    \forall A,B\in \SSS\text{ with }\hat{\Pr}\left((\hat{X_1},\hat{X_2})=(A,B)\right)>0\text{ it follows that }A\preceq_{\SSS} B. \label{eqn:monotoneCoupling}
  \end{align}
  A coupling is called monotone \whp (written $X_1\leq_{\text{\whp}} X_2$) if for some $c>1$ it satisfies
  \begin{align}
    \sum\limits_{A\not\preceq B}\hat{\Pr}\left((\hat{X_1},\hat{X_2})=(A,B)\right)\leq\frac{1}{n^c}. 
    \label{eqn:whpmonotoneCoupling}
  \end{align}
\end{definition}
Colloquially speaking, having a monotone coupling between two rumor spreading processes means that 
one process is at least as effective as the other one in every possible aspect.
More precisely, condition (\ref{eqn:monotoneCoupling}) says that if, in the joint distribution, there is a positive probability that process 1 informs exactly the nodes in $A$ and process 2 informs the nodes in $B$, then process 2 will (at least) inform all nodes which are informed by process 1. 
\mbox{Condition (\ref{eqn:whpmonotoneCoupling})} says that condition (\ref{eqn:monotoneCoupling}) holds with high probability.

The following theorem, Strassen's Theorem \cite{strassen:article, strassen:book}, shows an equivalence between stochastic dominance and the notion of monotone couplings.
\begin{theorem}[Strassen] 
\label{thm:Strassen}
The following are equivalent:
\begin{myenumerate}
\item $X_1$ \emph{stochastically dominates} $X_2$,
\item There exists a monotone coupling between $X_1$ and $X_2$ such that $X_1 \leq X_2$,
\item $\Pr\left(X_1\in F\right)\geq \Pr\left(X_2\in F\right)$ holds for every monotone set $F\subseteq \SSS$. \footnote{A set $F\subseteq \SSS$ is called monotone if $A\in F$ and $A\preceq_{\SSS}B$ implies $B\in F$.}
\end{myenumerate}
\end{theorem}

We want to show that $\bigO(\ddlogn)$ rounds of random \rpull stochastically dominate one round of \pull. This, however, is not possible as one can easily construct a graph in which some node $u$ is informed with probability $1$ in one round of \pull, but with probability less than $1$ in $\bigO(\ddlogn)$ rounds of \rpull.\footnote{Figure \ref{fig:onesteptight} in Appendix \ref{app:pictures} can be easily used to verify this.} Hence  a monotone coupling does not exist either. We therefore introduce the notion of highly probable stochastical dominance in analogy to the equivalencies from Strassen's Theorem.

\begin{definition}
  $X_1$ \emph{stochastically dominates} $X_2$ \emph{with high probability}, if there exists a coupling between $X_1$ and $X_2$ that is monotone with high probability.
\end{definition}

\subsection{W.h.p. Monotone Coupling between \boldmath \pull and \rpull}

\begin{theorem}
  \label{thm:rvpull}
  \Whp, for any set of informed nodes $S \subseteq V$, $T=\bigO\big(\frac{\Delta}{\delta} \log n\big)$ rounds of random $\rpull$ stochastically dominate a single round of $\pull$.
\end{theorem}
\begin{corollary}
If in a graph $G$ with initially informed nodes $S \subseteq V$ the \pull algorithm informs all nodes in $\TTT$ rounds with probability $p$, 
then the random \rpull algorithm informs all nodes in time $O\big(\TTT\cdot\frac{\Delta}{\delta}\cdot \log n\big)$ with probability $(1-o(1))p$.
\end{corollary}
By $\pushrpull$ we denote the combination of $\rpull$ with a simultaneous execution of the classic $\push$ protocol.
The restriction of a single node to answer only a limited number of requests does not limit the progress of the $\push$ algorithm when disseminating a rumor. Hence we deduce the following corollary.
\begin{corollary}
\Whp, for any set of informed nodes $S \subseteq V$, $T=\bigO\big(\frac{\Delta}{\delta} \log n\big)$ rounds of $\pushrpull$ stochastically dominate a single round of $\pushpull$.
\end{corollary}

To reduce dependencies between nodes which request from the same neighbor we introduce a new algorithm \vpull\ (virtual pull), which we let run for $T+1$ rounds and which, in any of those rounds, is strictly inferior to \rpull\ -- except for some rare cases that, \whp, do not arise. Note that \vpull\ is only introduced as a tool to analyze the algorithm \rpull; hence difficulties/impossibilities that arise in an actual implementation of \vpull\ are not relevant. 
The proof of Theorem \ref{thm:rvpull} is then split into two parts: 
\begin{enumerate}
  \item \mbox{Lemma \ref{lemma:rpullDominatesVpull}}: \Whp, $T$ rounds of \rpull\ stoch. dominate $(T+1)$ rounds of \vpull,
  \item Lemma \ref{lemma:vpullDominatesPull}: $(T+1)$ rounds of \vpull\ stochastically dominate one round of \pull.
\end{enumerate}
Then Theorem \ref{thm:rvpull} follows from the transitivity of the stochastical dominance relation.


By $\roundrpull$ we denote the (randomized) process \rpull\ which runs for $T$ rounds, by $\roundvpull$ we denote the process \vpull\ which runs for $T+1$ rounds and by $\roundpull$ we denote the process \pull\ which runs for one round only. The random variables $\Srpull$, $\Svpull$ and $\Spull$ denote the respective sets of nodes that are informed after the corresponding number of rounds. The processes $\roundrpull$, $\roundvpull$ and $\roundpull$ are not completely characterized by the random variables $\Srpull$, $S_{T+1}^{\vpull}$ and $S_1^{\pull}$ -- one has to include information about all requests and messages, that are sent by all nodes, to fully describe the random processes.
Nevertheless, to show the desired result, it is sufficient to find a monotone coupling where condition (\ref{eqn:monotoneCoupling}) and (\ref{eqn:whpmonotoneCoupling}), respectively, are fulfilled with regard to the subset relation of the set valued random variables $\Srpull$, $\Svpull$ and $\Spull$.

\subsubsection*{Definition of $\boldvpull$}
\label{subsection:badExecution}
An execution of \vpull consists of two phases. In the first phase nodes send tokens instead of the actual rumor and \whp nodes who have received a token in the first phase are informed at the end of the second phase. In an execution of \vpull we let $X_v(t)$ be the number of tokens which node $v$ has sent up to round $t$. In a specific round $t$ denote with $R_v(t)$ the set of nodes requesting from some informed node $v \in S$ and with $r_v(t)=|R_v(t)|$ its cardinality. $R_v$, $r_v$ and $X_v$ are random variables which describe certain properties of an execution of \vpull, where large values of $X_v$ or $r_v$ indicate the unlikely case in which the (strict) monotonicity of the coupling might break.

Let us also define \emph{weakly connected} nodes $u \in U$ as nodes for which $d_S(u)/d(u) \leq 1/2$
and \emph{strongly connected} otherwise. Let
\begin{equation*}
  K  =\Theta\left(\frac{\Delta}{\delta}+\log n\right), \quad T' =\bigO\left(\ddlogn\right) \quad \text{ and } \quad T=\Theta(T'), \, T \gg T'
\end{equation*}

\begin{definition}[Good, Bad Execution]
A $T$-round execution of \vpull is called a \emph{bad execution} if  for some $v\in V$ or $1\leq t\leq T$ it holds that $X_v(t)> K$ or $r_v(t)> K$, otherwise it is called a \emph{good execution}. 
\end{definition}
\hide{

}
First, we describe the algorithm informally. An execution of \vpull is split into two phases -- the first phase consists of $T$ rounds and the second phase of one round. In the first phase an uninformed node requests the rumor uniformly at random from one of its neighbors and an informed node $v$ decides with probability $\frac{r_v}{T'}$ whether to send out a token -- in which case it selects, uniformly at random, one of its incoming requests as destination for the token.
Nodes that get a token in those $T$ rounds, stop requesting from neighbors, but are still unable to forward any information to neighbors in consecutive rounds. In round $T+1$ the limit to the number of requests that can be served by an informed node is stripped away. Then, in case of a bad execution all actions from the first $T$ rounds are discarded and all uninformed nodes perform one round of \pull. In case of a good execution all uninformed strongly connected nodes perform one round of \pull and afterwards all nodes holding a token are being informed. If we assume that tokens are as valuable as the information itself, in each of the first $T$ rounds, \vpull differs from \rpull only in the fact that the selected incoming connection is established with probability $\frac{r_v}{T'}$ whereas it is established deterministically in \rpull. For an uninformed node $u\in U$, that chooses to request a neighbor $v \in S$, this normalizes the probability to get a token to $1/T'$, independent of the amount of other requesting nodes. 
Except for round $T+1$ this algorithm is clearly dominated by $\rpull$. 

A formal definition is given by the following pseudocode where the parameters  $K$ and $T'$ are defined as above. Note that the variables $X_v$, $R_v$, $r_v$, $BE_v$ and $BE$ can either be understood as random variables describing an execution of the \vpull algorithm or they can be updated directly in the algorithm as done below.
Except for \mbox{line \ref{vpull:oracleRequest}} which uses global knowledge \vpull can be seen as a distributed algorithm.
\begin{algorithm}[!htb] 
  \caption{One $(T+1)$-round execution of \vpull}
  \label{algo:vpull}
  \small
  
  \textbf{Input:} $K$ -- threshold for bad execution; $T'$ -- parameter to normalize probabilities\\
  \textbf{States:} \Sstate; \Ustate\\ 
  \textbf{Oracle knowledge:} $d_S(v)$ for every node $v$; $\badexec := \bigvee_{v \in V} \badexec_v$\\
  \begin{tabular}{@{}lll}
    \textbf{Variables: } &$R_v$ &  set of nodes requesting from $v$ in the corresp. round ($r_v \coloneqq |R_v|$)\\
    &$\badexec_v$ & boolean indicator for bad execution caused at node $v$\\
    &$\waiting$ & indicates whether a node will be informed after $T$ rounds
  \end{tabular}

  \begin{algorithmic}[1]
    \State $\badexec_v \gets \false$; $\waiting \gets \false$
    \For{$T$ rounds}
      
      \Switch{$\state$}
        \Case{\Ustate}
          \If{$\waiting = \false$}
            \State send request for rumor uniformly at random
            \If{$\msg = \ack$}
              \State $\waiting \gets \true$
            \EndIf
          \EndIf
        \EndCase
        \Case{\Sstate}
          \If{$r_v > K$ or $X_v > K$}
            \NoAlignComment{bad execution has been detected locally}
            \State $\badexec_v \gets \true$
          \Else
            \WithProb{$r_v/T'$} 
              \label{line:sendack}
              \State send $\ack$ to uniformly at random chosen node in $R_v \neq \emptyset$
              \State $X_v \gets X_v + 1$
            \EndWithProb
          \EndIf
        \EndCase
      \EndSwitch
    \EndFor
    \FullLineComment{Round $T+1$:}
    \State request $(\badexec, d_S(v))$ from global oracle \label{vpull:oracleRequest}
    \If{$\badexec = \true$}
      \NoAlignComment{bad execution has been detected globally}
      \State execute one round of \pull
      \NoAlignComment{i.e., informed nodes inform \emph{all} requesting neighbors}
    \Else
      \If{$d_S(v)/d(v)> 1/2$} \NoAlignComment{node is strongly connected}
        \State execute one round of \pull
      \ElsIf{$\waiting = \true$}
        \State $\state \gets \Sstate$\NoAlignComment{node learns rumor}
      \EndIf
    \EndIf
  \end{algorithmic}
\end{algorithm}


\subsubsection*{W.h.p. Monotone Coupling between \boldrpull\ and \boldvpull}
\label{section:couplingRpullVpull}
We generate first a coupling between \rpull\ and \vpull. In more layman terms imagine a (random) binary string $\sigma$ that contains all the information to generate either process in such a way that the informed nodes $S_{T+1}^\vpull$ are a subset of $S_T^\rpull$ for almost all strings $\sigma$; actually the probability that $\sigma$ is chosen in a way that $S_{T+1}^\vpull$ is not a subset of $S_{T}^\pull$ is less than $n^{-c}$. 

The coupling works in the following way. For each round $t \in \set{1, \dots, T+1}$ and each node $u$ we generate some random values $s_u(t)$, $s'_u(t)$ in $[0,1]$. If $u$ is uninformed (in either algorithm) and has not received a token at the beginning of round $t$ then $s_u(t)$ is used to determine which neighbor $u$ contacts, otherwise (if $u$ is informed) $s_u(t)$ is used to select to which requesting node (if any) a token or the information, respectively, is handed over. In \vpull, $s'_u(t)$ is hereby used to determine whether $v$ does send out any message at all, confer line \ref{line:sendack} from \Cref{algo:vpull}. Clearly, a node $u$ that is provided with a token in \vpull in any round $t \leq T$ is then also informed in \rpull. For round $T+1$ in \vpull the values $s_u(T+1)$ are used to simulate one round of \pull for any node that is required to do so, as stated in the \vpull\ algorithm.

We claim that, \whp, 
$s_u(T+1)$ is not used in the execution of \vpull for any node $u$ that does not get informed in $\roundrpull$, which thus implies that, \whp, $S_{T+1}^\vpull \subseteq S_T^\rpull$.

\begin{lemma}
  \label{lemma:rpullDominatesVpull}
  \roundrpull\ stochastically dominates  \roundvpull\ with high probability.
\end{lemma}

\begin{proof}
  Under the assumption that tokens are as valuable as the information itself we constructed a monotone coupling of $S_T^\vpull$ and $S_T^\rpull$. Now, it is sufficient to prove that in round $T+1$ of \vpull, \whp, no node is informed, that has not been informed in the $T$ rounds of \rpull:  If neither ever any value $r_v$ nor any $X_v$ exceeded $K$, then only strongly connected nodes simulate one round of \pull in round $T+1$ of the \vpull\ algorithm. We claim that each strongly connected node has been informed in the first $T$ rounds of \rpull. 


A strongly connected node $u \in U$ requests from an informed node $v \in S$ with probability at least $1/2$. In any given round due to Markov inequality with probability at least $1/2$ no more than $2\Delta/\delta$ nodes $u' \in U$  connect to $v$. The probability for $u$ to get informed under \rpull is thus at least $\frac{\delta}{8\Delta}$. Choosing $T=\bigO\big(\ddlogn\big)$ big enough and a union bound gives us that, \whp, all strongly connected nodes are informed in process $\roundrpull$.

  To conclude, we prove that \whp\ neither $r(v)$ nor $X(v)$ exceed $K$ for any node $v$ during an execution of \vpull. Let $1 < \kappa < c_{T,T'}$ be constants, $K':=\ddelta+\log n$ and $K=c_{T,T'} K'$. 
	

 \paraclose{\Whp, $\mathbf{r_v \leq K}$ in $\boldroundvpullT$ for all $\mathbf{v}$}For a fixed informed node $v$, in expectation, no more than $\ddelta$ nodes can request from $v$. Using a Chernoff bound 
 for a single round and a single node, $\Pr\big(r_v \geq \ddelta + \kappa \log n\big) \leq n^{-\Theta(\kappa)}$ holds. With a union bound over all nodes and all rounds and $\kappa$ large enough we obtain that, \whp, $r_v$ never exceeds $\kappa K'$ and therefore neither $K$. A union bound over all nodes concludes the proof.
	
 \paraclose{\Whp, $\mathbf{X_v\leq K}$ in \boldroundvpullT\ for all $\mathbf{v}$}For a fixed $v$, note that, \whp, in a single round no more than $\kappa K'$ nodes request from $v$, and therefore, $X_v$ is increased at most with probability $\kappa K'/T'$ in any round. Over $T$ rounds, in expectation, no more than $\kappa K' \frac{T}{T'}$ increments of $X(v)$ happen, and again a Chernoff bound gives us that $X_v$ does not exceed $2\kappa K' \frac{T}{T'}$ with high probability. Choosing $c_{T,T'} = 2\kappa \frac{T}{T'}$ and a union bound over all nodes concludes the proof.
%
\end{proof}

\subsubsection*{Stochastic Dominance between \boldvpull\ and \boldpull}
\label{section:couplingVpullPull}

In a single round of \pull\  a node $u\in U$ is informed with probability $\frac{d_S(u)}{d(u)}$, independently from which other nodes are informed. For $T+1$ rounds of \vpull\ we can show that a node is informed at least with the same probability and independently from 
which other nodes get informed, as claimed in the next lemma.
Afterwards, we prove that Lemma \ref{lemma:uInformedCondition} is sufficient to deduce the stochastic dominance of \roundvpull\ over \roundpull.
For $u\in U$ and random process $X$, let $\CCC_u^{X}$ be the set of all conditions of the type $v\in X$ or $v\notin X$ where $v\neq u$.
\begin{lemma}
\label{lemma:uInformedCondition}
In \roundvpull\ a node $u\in U$ is informed at least with probability $\frac{d_S(u)}{d(u)}$, independently from which other nodes are informed, i.e., for all sets of conditions $I\subseteq \CCC_u^{\vpull}$  and $J\subseteq \CCC_u^{\pull}$ with $\Pr(I), \Pr(J)>0$ the following holds
\begin{align}
 \Pr\left(\left.u\in S_{T+1}^{\vpull} \right| I\right)\geq \frac{d_S(u)}{d(u)}=\Pr\left(u\in S_{1}^{\pull}\right)= \Pr\left(\left.u\in S_{1}^{\pull} \right| J\right). \label{eqn:singleNode}
\end{align}
\end{lemma}
\begin{proof}
If $u\in U$ is strongly connected, the result holds because \vpull\ executes one round of \pull\ for $u$ in either way.
In a bad execution, \vpull\ executes one round of \pull\ for any uninformed node and the claim holds trivially. 
Thus assume that $u$ is weakly connected and we are in a good execution. 
Let $s=d_S(u)$ and $N_S(u)=\{v_1,\ldots,v_s\}$ be the neighbors of $u$ in $S$. We call a node $v \in N_S(u)$ \emph{busy} w.r.t. $u$ in round $t$ if it informs some node other than $u$. Let $y_t$ 
be the number of busy nodes in round $t$ w.r.t. $u$. 
In a good execution (which we denote by $\GoodExecution$), any node in $N_S(u)$ can inform at most $K$ nodes and hence there is the following constraint on the sum of all $y_t$'s
\begin{align}
\label{cond:yt}
\sum_{t=1}^Ty_t\leq s\cdot K.
\end{align}
 We can ignore conditions in $I$ corresponding to nodes which do not have a common neighbor with $N_S(u)\cup \{u\}$ because $u$ can only get the rumor directly through $S$. The only negative effect on the probability that $u$ gets informed by the conditions in $I$ can be captured by the number of busy nodes w.r.t. $u$. However, since the number of nodes which are informed per node in a good execution is small compared with $T$, there are sufficiently many rounds with sufficiently many non-busy nodes to inform $u$.
More precisely, if $u$ requests from a non-busy node it is informed at least with probability $\frac{1}{T'}$. Thus, the probability that $u$, conditioned on $I \wedge \GoodExecution$ with $\Pr(I \wedge \GoodExecution)>0$, is not informed is smaller or equal to (with $c=T/T'$)
\begin{align*}
&  \prod_{t=1}^T\left(1-\frac{s-y_t}{d(u)\cdot T'}\right)
 \leq \left(1-\frac{s\left(1-\frac{K}{T}\right)}{d(u)\cdot T'}\right)^T
\leq e^{-c\left(1-\frac{K}{T}\right)\frac{s}{d(u)}}
 \leq 1-\frac{d_S(u)}{d(u)}.
\end{align*}
The first inequality holds because under constraint (\ref{cond:yt}) the expression on the left hand side  is maximized for $y_t=\frac{s\cdot K}{T}$. The last inequality holds due to $\frac{s}{d(u)}\leq 1/2$, $c(1-K/T) \geq 2$ and the fact that $e^{-2x} \leq 1-x$ for any $x \in [0,1/2]$.
\end{proof}
The following result is due to Holley \cite{Holley74} and provides a sufficient criterion for stochastic dominance if the measures are chosen accordingly, e.g., as in the proof of Lemma \ref{lemma:vpullDominatesPull}.
\begin{theorem}[Holley Inequality, \cite{Holley74}]
\label{thm:holley}
Let $(\SSS,<)$ be a distributive lattice and let $\mu_1, \mu_2$ be measures on this lattice. The Holley criterion is satisfied if
\begin{align}
\mu_1(A\cap B)\mu_2(A\cup B)\geq \mu_1(A)\mu_2(B) \text{ holds for all } A,B\in \SSS. \label{holleyCriterion}
\end{align}
If the Holley criterion is satisfied for $\mu_1$ and $\mu_2$ then 
\begin{align}
\sum\limits_{A\in \SSS}\mu_1(A)f(A)\geq \sum\limits_{A\in \SSS}\mu_2(A)f(A) \text{ holds for all increasing functions } f:\SSS\rightarrow \RR. 
\end{align}
\end{theorem}

\begin{lemma}
  \label{lemma:vpullDominatesPull}
  \roundvpull\ stochastically dominates \roundpull. 
\end{lemma}
\begin{proof}
For the proof let $U$ be those uninformed nodes $u$ with  $0<d_S(u)<d(u)$ and consider the distributive lattice $(\SSS,\preceq_{\SSS})=(2^U,\subseteq)$. Every uninformed node $u$ which is not contained in this redefined $U$ has either no connection to $S$ at all, i.e., it is not informed in either process, or $d_S(u)=d(u)$ holds, i.e., it is informed with probability one in either process because also $\vpull$ executes one round of \pull for it. Hence it is sufficient to show stochastic domination of \Svpull\ over \Spull\ restricted to this redefined set $U$.
This choice of $U$ provides $0<\Pr(\Spull=A), \Pr(\Svpull=A)<1$ for all $A\in \SSS$ and we define the \emph{strictly positive} measures $\mu_1(F) :=\Pr\left(\Svpull\in F\right)$ and $\mu_2(F) :=\Pr\left(\Spull\in F\right)$ for $F\subseteq 2^U$.
 For $A$ in $\SSS=2^U$, $x\in U$ define $A^x:=A\cup \{x\}$ and $A_x:=A\backslash \{x\}$. The proofs of the follow claim  is based on Lemma \ref{lemma:uInformedCondition}.

\begin{restatable}[Quotient Rule]{claim}{quotientrule}
  \label{claim:quotientRule}
  \begin{align}
    \frac{\mu_1(A^{x})}{\mu_1(A_x)}\geq \frac{\mu_2(B^{x})}{\mu_2(B_x)}\text{ holds for all }A,B\in \SSS. \label{quotientRule}
  \end{align}
\end{restatable}
\begin{quotation}
At first note that for any $C \in \SSS$ it holds that
\begin{equation}
  \label{eq:claim:pull:constant}
  \Pr(\Spullz = C^x| \Spullz\in(\set{C^x}\cup \set{C_x}))=\Pr(x \in \Spullz) = \frac{d_S(x)}{d(x)}.
\end{equation}
This is true because in \pull\ $x$ is informed independently of what else is happening and due to the fact that we already condition on $\Spullz$ being either $C^x$ or $C_x$, hence the probability of $\Spullz$ being $C^x$ depends solely on $x$ being informed.
Second, recall that $\mu_1$ and $\mu_2$ are strictly positive measures, so $\mu_1(C),\mu_2(C)>0$ for every $C \in \SSS$, even $C=\emptyset,U$.

Let $A,B\in \SSS$ and $x\in U$. Then Lemma \ref{lemma:uInformedCondition} implies
\renewcommand{\qedsymbol}{\loweredqedbox}
\begin{align*}
  \pushQED{\qed}
  &\Rightarrow & \Pr(\Svpullz = A^x | \Svpullz\in (\set{A^x}\cup \set{A_x}))& \geq \Pr(\Spullz = A^x| \Spullz\in(\set{A^x}\cup \set{A_x}))\\
  & & \stackrel{(\ref{eq:claim:pull:constant})}{=}
  \Pr(\Spullz= B^x| \Spullz\in(\set{B^x}\cup \set{B_x}))\\
  & \Rightarrow  &\frac{\mu_1(A^x)}{\mu_1(A^x)+\mu_1(A_x)} &\geq \frac{\mu_2(B^x)}{\mu_2(B^x)+\mu_2(B_x)}\\
  & \Rightarrow  &\frac{\mu_1(A^x)+\mu_1(A_x)}{\mu_1(A^x)} &\leq \frac{\mu_2(B^x)+\mu_2(B_x)}{\mu_2(B^x)}\\
  & \Rightarrow  &\frac{\mu_1(A_x)}{\mu_1(A^x)} &\leq \frac{\mu_2(B_x)}{\mu_2(B^x)}\\
  & \Rightarrow  &\frac{\mu_1(A^x)}{\mu_1(A_x)}& \geq \frac{\mu_2(B^x)}{\mu_2(B_x)}.
  \qedhere
  \popQED
\end{align*}
\end{quotation}

\begin{restatable}{claim}{claimholley}
  \label{claim:holley}
  The quotient rule (\ref{claim:quotientRule}) implies that the Holley criterion is satisfied for $\mu_1$ and $\mu_2$.\footnote{The proof is adapted from \cite[chapter 2, page 24]{Grimmett06}.}
\end{restatable}
\begin{quotation}

Let $A,B\in \SSS$ and $C:=A\backslash B=\{c_1,\ldots,c_r\}$. The Holley criterion is trivially fulfilled if $A \subseteq B$; hence assume otherwise, which implies $r \geq 1$.
Write $C_s:=\{c_1,\ldots,c_s\}$ for $1 \leq s\leq r$. By a telescoping argument we obtain the following.
\renewcommand{\qedsymbol}{\loweredqedbox}
\begin{align*}
  \pushQED{\qed}
\frac{\mu_1(A\cup B)}{\mu_1(B)}& =\frac{\mu_1(B\cup C)}{\mu_1(B\cup C_{r-1})}\cdot \frac{\mu_1(B\cup C_{r-1})}{\mu_1(B\cup C_{r-2})}\cdot \cdots \cdot \frac{\mu_1(B\cup C_1)}{\mu_1(B)}
\intertext{Applying (\ref{quotientRule}) to each fraction we obtain}
& \geq \frac{\mu_2((A\cap B)\cup C)}{\mu_2((A\cap B)\cup C_{r-1})}\cdot \frac{\mu_2((A\cap B)\cup C_{r-1})}{\mu_2((A\cap B)\cup C_{r-2})}\cdot \cdots \cdot \frac{\mu_2((A\cap B)\cup C_1)}{\mu_2((A\cap B))}\\
& = \frac{\mu_2(A)}{\mu_2(A\cap B)}.\qedhere
\popQED
\end{align*}
\end{quotation}
The proof of Lemma \ref{lemma:vpullDominatesPull} then follows with  Claim \ref{claim:holley}, Theorem \ref{thm:holley} and the definition of the expected value of an increasing function $f:\SSS\rightarrow \RR$.
\end{proof}

\begin{proof}[Proof of Theorem \ref{thm:rvpull}]
  The proof is a direct combination of Lemma \ref{lemma:rpullDominatesVpull} and Lemma \ref{lemma:vpullDominatesPull}.
\end{proof}


\subsection{The Round-by-Round Analysis is Tight}

\begin{lemma}
  \label{lemma:onesteptight}
  The time bound $T=\bigO(\ddlogn)$ from Theorem \ref{thm:rvpull} is tight.
\end{lemma}

\begin{proof}
  In order for $T$ rounds of \rpull to quasi-dominate $1$ round of \pull, any node $v$ must get informed in $T$ rounds of \rpull with at least the same probability as within one round of \pull (or \whp, if the latter probability equals one). We construct a graph $G$ for which at least $T=\Omega(\ddlogn)$ rounds of \rpull are necessary to guarantee this. A picture depicting the graph can be found in Appendix \ref{app:pictures}.


We partition set $V$ into $V=A \dcup B \dcup T_{1,1} \dcup \dots \dcup T_{k^2,k^2}$, where $A=\set{a_1, \dots, a_{k^2}}$, $B=\set{b_1,\dots,b_{k^2}}$ and $k=n^{1/5}$. $A$ and $B$ form a complete bipartite graph with edges running between $A$ and $B$. For each $j \in [k^2]$, node $b_i$ is connected to \emph{one} node $t_{i,j} \in T_{i,j}$. Each $T_{i,j}$ forms a complete graph of size $k$.
In this graph, $\delta=k-1$ (acquired in $T_{i,j}$) and $\Delta=2k^2$ (nodes in $B$), and therefore $\Delta/\delta \in \Theta(k)$.
The total size of the graph is $|V| = n + o(n)$. Initially, we let $S_0 = B$.

In this graph, within one round of \pull, all nodes of $A$ are informed with probability $1$. 
Now, consider the same graph after $m \leq k^2/2$ rounds of \rpull and let us assume that some node $a \in A$ is still uninformed. It requests in this round from some node $b_i$. Let $X_i$ be the number of requests at $b_i$. Within $m$ rounds, each node $b_i$ managed to inform at most $m$ of its neighbors from $\NB_i:=\set{t_{i,1}, \dots, t_{i,m}}$. Since $m\leq k^2/2$, at least half of all nodes in $\NB_i$ are still uninformed and thus, since they have degree $k$, 
$\E[X_i]\geq k/2$.
Applying Chernoff, we get that \whp, 
$X_i \geq k/4$.
In this scenario for $a$ the probability to be chosen over one of its competitors is at most $4/k$, regardless of $m$, and therefore, 
$\Pr(a \in S^\rpull_m) \leq 1-\big(\frac{4}{k}\big)^m$.
For this to exceed $1-1/n$, $m$ has to be in $\Theta\big(\ddlogn\big)$.
\end{proof}


\section{Conclusions}

Lemma \ref{lemma:onesteptight} and Theorem \ref{thm:rvpull} show that to simulate one round of \pull, $\Theta\big(\ddlogn\big)$ rounds of \rpull are required. However, in case one wants stochastical dominance (\whp) over $\TTT > 1$ rounds of \pull, the lower bound proof of Lemma \ref{lemma:onesteptight} does not hold anymore. We believe that for $\TTT = \Omega(\log n)$, on any graph $G$ and any set of initially informed nodes $S \subseteq V$, $\bigO\big(\TTT\big(\ddelta + \log n\big)\big)$ 
or maybe even $\bigO\big(\TTT\big(\ddelta\big)+ \log n\big)$
rounds of \rpull suffice to stochastically dominate $\TTT$ rounds of
\pull. That proving this assumption might be a challenging task is underlined by a similar conjecture in \cite{AcanMehrabianWorlmald14}, in which the authors do a coupling of synchronous and asynchronous \pushpull. They obtain a similar multiplicative $\bigO(\log n)$ factor and also conjecture 
that it can be improved to an additive $\bigO(\log n)$ term. 

A possible alternative restriction of the \pushpull protocol could
be given by the following algorithm. In each round, every node requests from an outgoing
neighbor chosen uniformly at random. At each node, one of the incoming
requests is chosen (e.g., uniformly at random) and a connection to the
requesting node is established. Finally, over all established links
between an informed and an uniformed node, the uninformed node learns
the rumor. Note that unlike in the restricted \pushpull variant
described in our paper, here, also two informed nodes or two
uninformed nodes could be paired. Such a \pushpull variant can be
analyzed in an analogous way to our analysis of the \rpull protocol
and it can be shown that $\bigO\big(\frac{\Delta}{\delta} \log n\big)$
rounds of this algorithm stochastically dominate a single round of the
regular \pushpull protocol.


\bibliographystyle{habbrv}
\bibliography{references}


\appendixpage
\appendix
\section{Pictures}
\label{app:pictures}

\subsection{Picture for Lemma \ref{lemma:separation}, Figure \ref{fig:separation}}
\label{section:pictureForSeparation}
\begin{figure}[htb]
  \centering
  \includegraphics[width=1\textwidth]{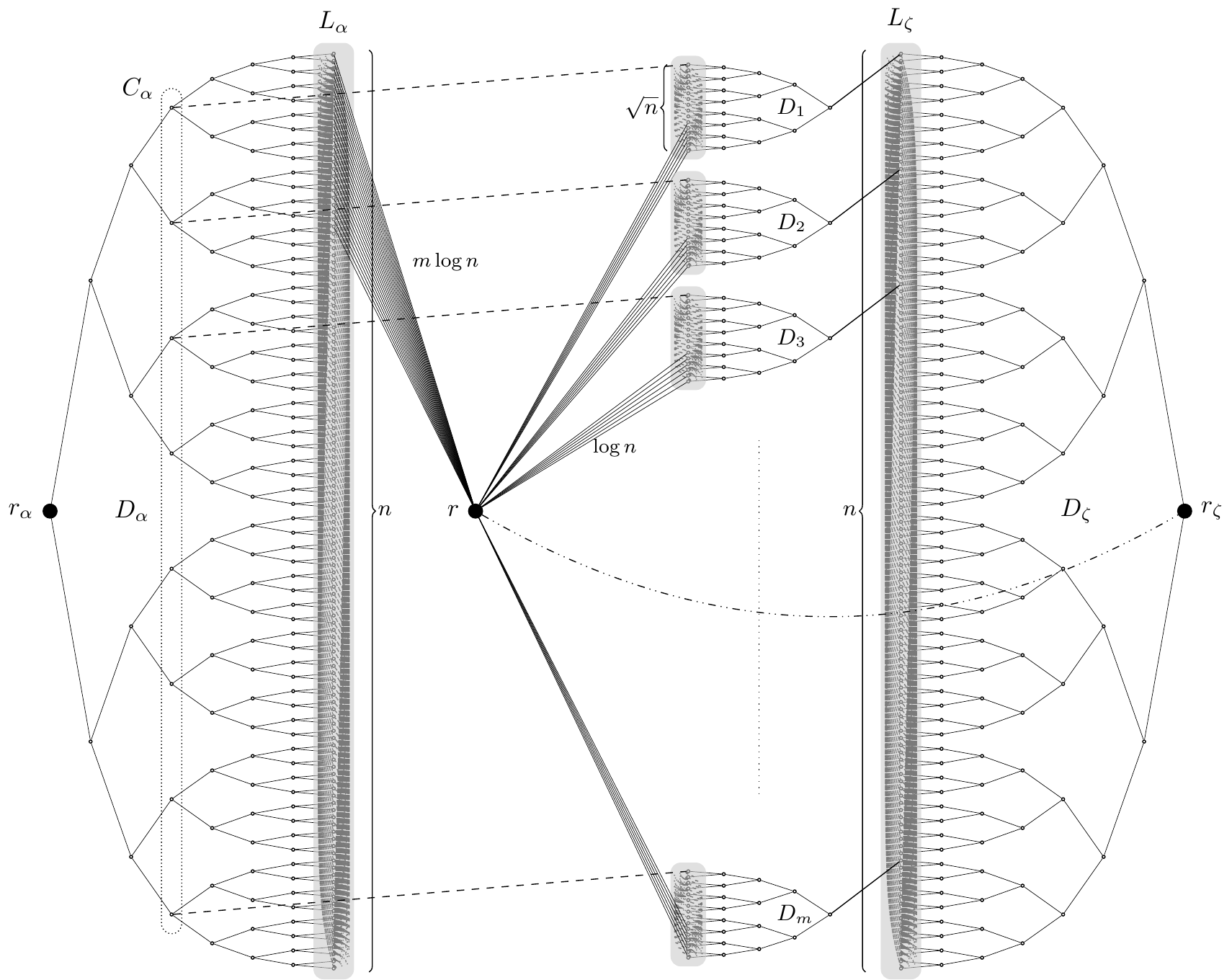}
  \caption{Picture for Lemma \ref{lemma:separation}, proving that random and adversarial \rpull have exponentially different running times on general graphs. Grey areas indicate fully connected parts of the graph.}
  \label{fig:separation}
\end{figure}
  	In random \rpull the node $r_{\zeta}$ learns the rumor within $\bigO(\log^2n)$ rounds and can spread the information through the graph in polylogarithmic time. 
		
		In adversarial \rpull the adversary prevents $r_{\zeta}$ from learning the rumor by always disseminating the rumor to one of the requesting nodes of in $D_1,\ldots,D_m$ in every round. We can show that the number of informed $D_i$s grows slowly and hence such requests exist \whp as long as no node in $D_\zeta$ is informed. Also, with only few $D_i$s informed, due to their high degrees, leaf nodes in $L_\zeta$ are unlikely to request from a $D_i$ containing the rumor, and hence the progress of rumor propagation is stalled.

\subsection{Picture for Lemma \ref{lemma:onesteptight}, Figure \ref{fig:onesteptight}}
\begin{figure}[htb]
  \centering
  \includegraphics[width=0.85\textwidth]{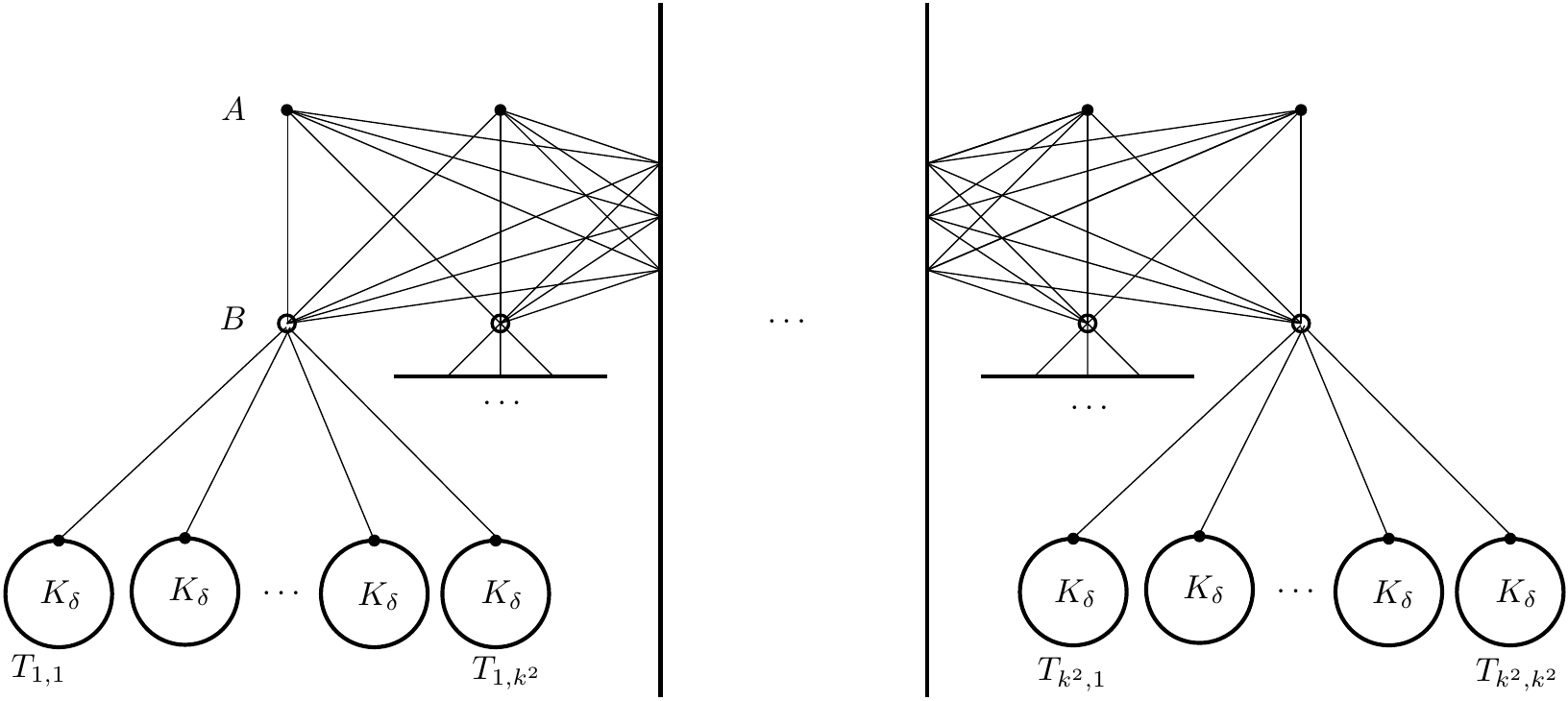}
  \caption{Picture for Lemma \ref{lemma:onesteptight}, proving that $\Omega(\ddlogn)$ rounds of \rpull are necessary to simulate one round of \pull.}
  \label{fig:onesteptight}
\end{figure}
  All not filled circular nodes in the bipartite graph in the top ($B$) have the rumor at the start of the execution. In one round of \pull all filled nodes in the bipartite graph ($A$) learn the rumor with probability one. In random \rpull $\Omega(\frac{\Delta}{\delta}\log n)$ rounds are necessary to inform these nodes, because 
for each of the nodes in $A$ the probability to be informed in one round of \rpull is in $\Theta(1/k)$. This is due to the high number of requests from nodes in $T_{i,j}$ each round to the informed nodes in $B$.


\section{Chernoff Statement for Lemma \ref{lemma:treefull}}
\label{appendix:chernoff}
\begin{lemma}
\label{thm:geometricchernoff}
 Let $X_1,X_2, \dots, X_n$ be independent geometric random variables with $X_i\sim Geo(p_i)$ for \mbox{$i=1,\ldots,n$} and \mbox{$0<p_1 \leq p_2 \leq ... < 1$}. Let $X=\sum_{i=1}^n X_i$.
  Then
\begin{align}
\Pr\left(X>3(\mu+t)\right)\leq e^{-\frac{p_1}{2}\mu-p_1t}.
\end{align}
\end{lemma}

\begin{proof}
Let $\gamma=-\ln(1-\frac{p_1}{3})>0$. Because of $p_1\leq p_i$ this implies $e^{\gamma}\cdot(1-p_i)<1$ for all $i=1,\ldots,n$. We need this condition at $(*)$ in the proof of the following claim.
\begin{claim}\label{claimChernoff}
$\E[e^{\gamma X_i}]\leq 1+\frac{p_1}{2p_i}.$
\end{claim}
\begin{quotation}
With a straight forward calculation one obtains
\renewcommand{\qedsymbol}{\loweredqedbox}
\begin{align*}
  \E\left[e^{\gamma X_i}\right] & = \sum\limits_{k=1}^{\infty}\Pr\left(X_i=k\right)e^{\gamma}
  =p_ie^{\gamma}\sum\limits_{k=1}^{\infty}\left((1-p_i)e^{\gamma}\right)^{k-1}\\
  &\stackrel{(*)}{=} \frac{p_ie^{\gamma}}{1-(1-p_i)e^{\gamma}}
  =\frac{p_i}{e^{-\gamma}-1+p_i}\\
  & = 1+ \frac{1-e^{-\gamma}}{e^{-\gamma}+p_i-1}
  = 1+\frac{1-(1-\frac{p_1}{3})}{1-\frac{p_1}{3}+p_i-1}\\
  & = 1+ \frac{\frac{p_1}{3}}{p_i-\frac{p_1}{3}}
  \leq 1+ \frac{p_1}{2p_i}.
  \qedhere
\end{align*}
\end{quotation}
Since $x\mapsto e^{\gamma x}$ is an increasing function for $\gamma>0$ we obtain
\begin{align*}
  \Pr\left(X>3(\mu+t)\right) & =\Pr\left(e^{\gamma X}>e^{\gamma 3(\mu+t)}\right) & \text{(Markov)}\\
  & \leq e^{-3\gamma(\mu+t)}\prod\limits_{i=1}^n\E\left[e^{\gamma X_i}\right] & \text{(claim 1)}\\
  & \leq e^{-3\gamma(\mu+t)}\prod\limits_{i=1}^n\left(1+\frac{p_1}{2p_i}\right) & (1+x\leq e^{x}, x\in\RR)\\
  & \leq \left(1-\frac{p_1}{3}\right)^{3(\mu+t)}e^{\sum_{i=1}^n\frac{p_1}{2p_i}} & \left(\mu=\sum_{i=1}^n\frac{1}{p_i}\right)\\
  & = \left(1-\frac{p_1}{3}\right)^{3(\mu+t)}(e^{\frac{p_1}{2}\mu}) & (1-x\leq e^{-x}, x\in\RR)\\
  & \leq e^{p_1(\mu+t)+\frac{p_1}{2}\mu}= e^{-\frac{p_1}{2}\mu-p_1t},
\end{align*}
which proves the actual result.\qed
\end{proof}





\end{document}